\documentclass[letterpaper, 10 pt, conference]{ieeeconf}
\IEEEoverridecommandlockouts
\usepackage[T1]{fontenc}
\usepackage[utf8]{inputenc}
\usepackage{comment}
\usepackage{amsmath,amssymb,amsfonts}
\usepackage{algorithmic}
\usepackage{graphicx}
\usepackage{textcomp}
\usepackage{xcolor}
\usepackage{hyperref}
\usepackage{mathtools}
\usepackage{soul}
\usepackage{afterpage}
\usepackage{url}

\makeatletter
\let\NAT@parse\undefined
\makeatother

\newtheorem{theorem}{Theorem}
\newtheorem{definition}{Definition}
\newtheorem{lemma}{Lemma}

\newtheorem{assumption}{Assumption}
\newtheorem{corollary}{Corollary}%[theorem]

\title{
On the utilization of Macroscopic Information \\ for String Stability of a  Vehicular Platoon
}%Is String Stability able to handle \\ Macroscopic Information?

\author{Marco Mirabilio, Alessio Iovine, Elena De Santis, Maria Domenica Di Benedetto, Giordano Pola% <-this % stops a space
\thanks{Marco Mirabilio, Elena De Santis, Maria Domenica Di Benedetto, Giordano Pola are with the Department of Information Engineering, Computer Science and
Mathematics, Center of Excellence DEWS, University of L’Aquila. (e-mail:marco.mirabilio@graduate.univaq.it, \{elena.desantis,mariadomenica.dibenedetto,giordano.pola\}@univaq.it)
}%
\thanks{Alessio Iovine is with the Electrical Engineering and Computer Sciences (EECS) Department at UC Berkeley, Berkeley, USA. E-mail: alessio@berkeley.edu, alessio.iovine@ieee.org.
}%
}
\begin{document}

%%%%%%%%%%%%%%%%%%%%%%%%%%%%%%%%%%%%%%%%%%%%%%%%%%%%%%%%%%%%%%%%%%%%%%%%%%%%%
%%%%%%%%%%%%%%%%%%%%%%%%%%%%%%%%%%%%%%%%%%%%%%%%%%%%%%%%%%%%%%%%%%%%%%%%%%%%%
\maketitle

\begin{color}{black}
\begin{abstract}
    The use of macroscopic information for the control of a vehicular platoon composed of autonomous vehicles is investigated. A mesoscopic control law is provided, and String Stability is proved by Lyapunov functions and Input-to-State Stability (ISS) concepts. Simulations are implemented in order to validate the controller and to show the efficacy of the proposed approach for mitigating traffic oscillations. 
    
\end{abstract}
\end{color}

%\begin{IEEEkeywords}

%\end{IEEEkeywords}
%\vspace{-10pt}

%%%%%%%%%%%%%%%%%%%%%%%%%%%%%%%%%%%%%%%%%%%%%%%%%%%%%%%%%%%%%%%%%%%%%%%%%%%%%
%%%%%%%%%%%%%%%%%%%%%%%%%%%%%%%%%%%%%%%%%%%%%%%%%%%%%%%%%%%%%%%%%%%%%%%%%%%%%
\section{Introduction}

Nowadays, Vehicle-to-Infrastucture (V2I) and Vehicle-to-Vehicle (V2V) communication technologies are a reality in the smart transportation domain (see \cite{Uhlemann2017VTM}), and their utilization in Cooperative Adaptive Cruise Control (CACC) is widely expected  to improve traffic conditions (see \cite{Feng2019ARC}, \cite{Besselink2017TAC}, \cite{Johansson2015CSM}). Indeed, traffic jamming transition has been shown to strongly depend on the amplitude of fluctuations of the leading vehicle (see \cite{Nagatami2000PRE}), and interconnected autonomous vehicles are sensed to reduce stop-and-go waves propagation and traffic oscillations via the concept of String Stability  \cite{Feng2019ARC}, \cite{Swaroop1996StringStability},  \cite{Piccoli2018TRC}, \cite{DeSantis2006}.  %However, the considered limit situation of full penetration of CACC driven autonomous vehicles is far to be reached. Due to the necessity to compensate missing information or information propagation from human-driven or simple ACC vehicles, the transition requires new formulations, mathematical modeling, and control analysis \cite{Bayen2018ICRA}.
String Stability relies on the idea that disturbances acting on an agent of the cluster should not amplify backwards in the string. 
%\textcolor{black}{Recently, the possibility to extend the notion of String Stability to the general class of networks of interconnected systems has been investigated; in \cite{Besselink2018CSL} the authors introduce the \textit{scalable} String Stability.} 
In the case of vehicular platooning, disturbances may be due to reference speed variation, external inputs acting on each vehicle, wrong modeling, etc. Several cases of information sharing have been considered for each leader-follower interaction, but a common characteristic is that some microscopic variables are always shared among the whole platoon, e.g. the acceleration of the platoon's leading vehicle (see \cite{Swaroop1996StringStability}) or its desired speed profile (see \cite{Besselink2017TAC}). It needs a V2V communication among the whole platoon, or a V2I bidirectional exchange of information. %Also, different spacing policies have been investigated with the target to  
%It results complicated to be actuated in a mixed traffic situation where not all vehicles communicate to each other.% and some can be human-driven.
%It imposes a high communication rate, as the 
This paper analyses the benefits of the information propagation in a String Stability framework using both microscopic and macroscopic information for control purposes.
%Microscopic variables as acceleration, speed and position of the leading vehicle, when transmitted, are considered only in each leader-follower situation. 
Each follower is here considered to correctly measure the distance and speed of its leading vehicles, using for example radar and LIDAR. The leader acceleration is communicated only to its follower. To improve control performance, macroscopic information is supposed to be obtained and communicated either from the road infrastructure (V2I) or from the whole platoon (V2V). \textcolor{black}{Both technologies have strengths and weaknesses. For example, V2V technology requires the macroscopic information to be propagated through the vehicles, and possibly estimated in a distributed manner by each one. On the other hand, V2I technology may provide a more reliable information at the cost of allocating several sensors along the way and computing the quantities in a centralized manner, implying a high computation request to the central computer. A thorough analysis of pros and cons of the two communication typologies is out of the scope of the present paper.} 

We target a platoon composed by autonomous vehicles implementing CACC, but the framework we propose is suitable for including autonomous vehicles implementing simple ACC or even human-driven vehicles as part of the platoon. %, to receive macroscopic information either from the road infrastructure (V2I) or from the whole platoon (V2V) and to receive acceleration information by its predecessor (as in \cite{Ploeg2014TAC}).
% Each autonomous vehicle is here supposed to correctly measure the distance and speed of its leading vehicles, using for example radar and LIDAR, and to receive only macroscopic information from the road infrastructure (V2I) or from the whole platoon (V2V) \textcolor{red}{while receiving acceleration information affected by an error by its leading vehicle. E" ANCORA IL CASO?}. 
The framework we propose is based on sharing macroscopic quantities along the platoon. %This is used to perform a time-varying spacing policy for each leader-follower interaction, where quantities as distance or relative speed can be either communicated or calculated. No sharing of single vehicle's microscopic quantities is mandatory. Indeed, macroscopic information
%
%
%The received macroscopic information can be used in the controller and/or in the adopted spacing policy.
The use of those quantities aims at increasing the ability of each car-following situation to counteract the disturbances by providing an anticipatory behaviour capable to absorb traffic jam. The idea of using macroscopic quantities, mainly the density, for microscopic traffic control has already been introduced in the literature, resulting in a mesoscopic modeling. In \cite{treiber_2013_book}, \cite{Zhu2016TITS} and in the references therein, the focus is on simulation aspects and real data analysis. Several works are now focusing on a mesoscopic modeling for traffic control purposes (see \cite{Johansson2019ECC}, \cite{Ferrara2018ITSC} \cite{DeSchutter2017TRC}, \cite{Iovine2015NAHS}, \cite{Iovine2015ADHSofficial}).

%However, these results focus on simulations aspects or real data analysis, and no rigorous analysis is available in a String Stability framework. In order to better utilize the available information and to improve the effects of autonomous vehicles on the whole traffic flow, several works are now focusing on a mesoscopic modeling of the traffic control framework (see \cite{Johansson2019ECC}, \cite{Ferrara2018ITSC} \cite{DeSchutter2017TRC}, \cite{Iovine2015NAHS}). 

The controller we present in this paper considers  macroscopic information and ensures Asymptotic String Stability. 
%that increases traffic throughput while reducing perturbation effects along the platoon. The proposed controller is in the control framework of the ones in \cite{Swaroop1996StringStability} and \cite{Besselink2017TAC}, while t
The adopted nonlinear spacing policy relies on the family of nonlinear spacing strategies introduced in \cite{Yanakiev1995CDC} and \cite{Zhang2005CDC}. 
Similarly to \cite{Besselink2017TAC}, the result is obtained through an inductive method exploiting Input-to-State Stability (ISS). The main difference is that 
%ISS is ensured among each leader-follower situation in \cite{Besselink2017TAC}
\textcolor{black}{ISS is ensured with respect to the leader-follower situation and the ahead vehicles of each predecessor.}  
%as shown by (\ref{eq:closed_loop_platoon_ISS_ineq2}). According to (\ref{eq_delta_forChi_bounded}), 
% \textcolor{blue}{The proposed solution does not depend on the number of vehicles along the platoon.}
Simulations show the improvements on the whole traffic throughput producing an anticipatory behaviour and oscillations reduction, and providing a better \textcolor{black}{transient harmonization} while maintaining String Stability properties.

%Two controllers are proposed, with different spacing policies. The first one considers a constant spacing policy, and macroscopic information is used to ensure String Stability. Instead, 
%The proposed controller considers a macroscopic information based time varying spacing policy that increases traffic throughput while reducing perturbation effects along the platoon. %The proposed controller is in the control framework of the ones in \cite{Swaroop1996StringStability} and \cite{Besselink2017TAC}, while t
%The nonlinear spacing policy relies on the family of nonlinear spacing strategies introduced in \cite{Yanakiev1995CDC} and \cite{Zhang2005CDC}. Simulations show the improvements on the whole traffic throughput producing an anticipatory behaviour and oscillations reduction, and providing a better transient harmonization while maintaining String Stability properties.

The paper is organized as follows. Section \ref{sct:modeling} introduces the considered framework, while Section 
\ref{sct:control_tools} the needed control tools. Control laws are derived and stability analysis is performed in Section \ref{sct:mesoscopic_control}. Simulations are carried out in Section  \ref{sct:simulations}. Some concluding remarks are outlined in Section \ref{sct:conclusion}.

\vspace{0.2cm}
\textbf{Notation - } $\mathbb{R}^+$ is the set of non-negative real numbers. For a vector $x\in\mathbb{R}^n$, $|x|=\sqrt{x^T x}$ is its Euclidean norm. The $\mathcal{L}_\infty$ signal norm is defined as $|x(\cdot)|_\infty^{[t_0,t]} = \sup_{t_0\leq\tau\leq t}|x(\tau)|$. %The $\infty$ norm of a vector is denoted by $|x|_\infty = \max_{i=1...n}|x_i|$. If a different norm is used, it is indicated by a subscript (e.g. $|x|_p$ denotes the generic $p$ norm). 
We refer to \cite{B_khalil_2002} for the definition of Lyapunov functions, and functions $\mathcal{K}$, $\mathcal{K}_\infty$ and  $\mathcal{KL}$.

%%%%%%%%%%%%%%%%%%%%%%%%%%%%%%%%%%%%%%%%%%%%%%%%%%%%%%%%%%%%%%%%%%%%%%%%%%%%%
%%%%%%%%%%%%%%%%%%%%%%%%%%%%%%%%%%%%%%%%%%%%%%%%%%%%%%%%%%%%%%%%%%%%%%%%%%%%%
\section{Modeling and String Stability definitions}\label{sct:modeling}

%%%%%%%%%%%%%%%%%%%%%%%%%%%%%%%%%%%%%%%%%%%%%%%%%%%%%%%%%%%%%%%%%%%%%%%%%%%%%
\subsection{Platoon modeling}
We consider a cluster of $N+1$ vehicles, $N\in\mathbb{N}$, proceeding in the same direction on a single lane road, as in Fig.\ref{fig:reference_framework}. We make the following 
\begin{assumption}
    All the vehicles are equal, with the same length $L\in\mathbb{R}^+$ and have the common goal of maintaininig a strictly positive distance among them, while keeping the same speed.
\end{assumption}
We denote with $i = 0$ the first vehicle of the platoon and with $\mathcal{I}_N = \{1,2,...,N\}$ the set of follower vehicles. The set including all the vehicles is $\mathcal{I}_N^0 = \mathcal{I}_N \cup \{0\}$.
\begin{figure}
    \centering
    \includegraphics[width = 1\columnwidth]{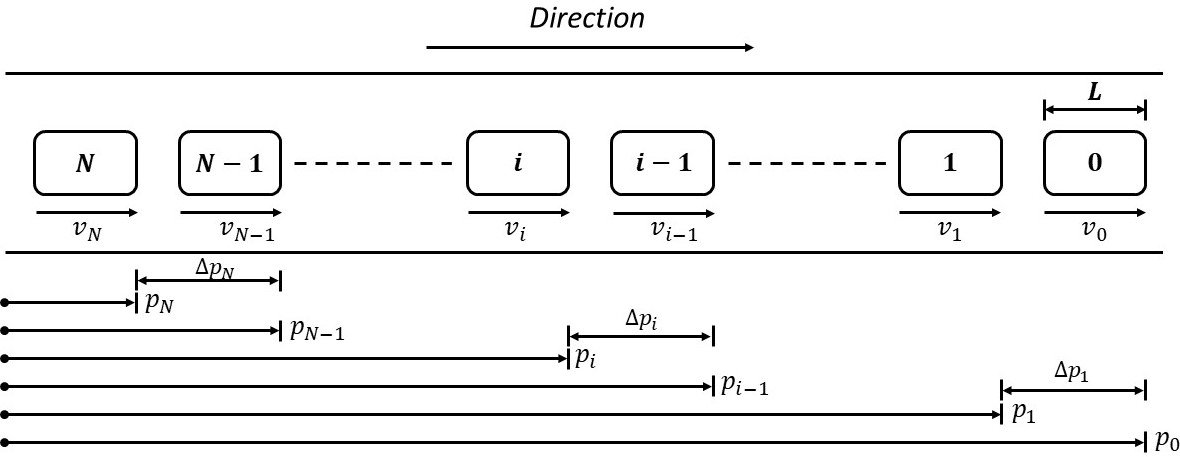}
    \caption{Reference framework.}
    \label{fig:reference_framework}
    \vspace{-10pt}
\end{figure}
Similarly to \cite{Besselink2017TAC}, each vehicle $i\in\mathcal{I}_N^0$ is assumed to satisfy the following longitudinal dynamics:
\begin{equation}\label{eq:generic_longitudinal_dynamics}
    \begin{array}{l}
        \dot{p}_i = f_p(\xi_i) \\
        \dot{\xi}_i = f_\xi(\xi_i)+g_\xi(\xi_i)u_i
    \end{array}
\end{equation}
where $p_i\in\mathbb{R}^+$ is the position of vehicle $i$, $v_i=\dot{p}_i$ ($0< v_i\leq v_{\max}, \ v_{\max}\in\mathbb{R}^+$) is its velocity and the acceleration $u_i$ is the vehicle control input. Variable $\xi_i\in\mathbb{R}^{n-1}$ represents the remaining dynamics of the vehicle, such as actuators dynamics.  
%(the acceleration). 
%Since the reaction delay has only slight quantitative influence on the oscillation growth pattern (see \cite{Treiber2016role}), no delays are considered here either for the reaction time or the communication time.

According to \cite{Swaroop1996StringStability} and \cite{Xiao2011StringStability}, the dynamics in (\ref{eq:generic_longitudinal_dynamics}) can be simplified:%\textcolor{red}{inserire ulteriore commento}
\begin{color}{black}
\begin{equation}\label{eq:longitudinal_dynamics}
     \dot{p}_i = v_i, \quad \dot{v}_i = u_i
\end{equation}
\end{color}
\noindent
\hspace{-10pt} where $\xi_i=v_i$, the functions $f_p(\xi_i)=v_i$, $f_\xi(\xi_i)=0$ and $g_\xi(\xi_i)=1$; $u_i$ is the acceleration of vehicle $i$ ($|u_i|\leq a_{\max}, a_{\max}\in\mathbb{R}^+$). \textcolor{black}{The introduced double integrator model is widely used in literature for string stability analysis purposes. Moreover, field experiments adopting control laws developed with respect to this model have shown satisfactory behaviors \cite{Swaroop2019Review}.} %\textcolor{red}{io non metterei la parte in blu. I riferimenti gia ci sono, che affermano cio' che scriviamo.}
\noindent
To provide a  global description of the platoon, we adopt the leader-follower model that describes the inter-vehicular interaction  (see \cite{Ploeg2014TAC}, \cite{Mirabilio2020ECC}). We define the state of each vehicle $i\in\mathcal{I}_N^0$ as
\begin{equation}\label{eq:state_vector_x}
    x_i = \left[ p_i \ \ v_i \right]^T
\end{equation}
and the state of each car-following situation among the leading vehicle $i-1$ and the following one $i$ as
\begin{equation}\label{eq:state_vector_chi}
    \chi_i = x_i-x_{i-1} 
    = \left[
        \begin{array}{c}
            \Delta p_i    \\
            \Delta v_i
        \end{array}
    \right] =  \left[
        \begin{array}{c}
            p_i-p_{i-1} \\
            v_i-v_{i-1}
        \end{array}
    \right].
\end{equation}
Positions, speed and acceleration of each leading vehicle are supposed to be perfectly known, either measured or communicated to the following one. Consequently, the obtained dynamical model is:
\begin{equation}\label{eq:car_following_dynamics_i}
   \dot{\chi}_i = \left[
    \begin{array}{c}
      \Delta \dot{p}_i    \\
      \Delta \dot{v}_i
    \end{array}
    \right] =  \left[
    \begin{array}{c}
         \Delta {v}_i \\
         u_i-u_{i-1}
    \end{array}
   \right], \quad i\in\mathcal{I}_N,
\end{equation}
or, shortly,
\begin{align}\label{eq:car_following_dynamics_i_short}
    \dot{\chi}_i  = f(\chi_i,u_i,u_{i-1}), \quad i\in\mathcal{I}_N.
\end{align}
%
% where $|d_i|\leq d_{\max}$ is assumed to be a bounded disturbance. 

% As $u_{i-1}$ is bounded, the assumption that it is bounded by a constant as well is done, i.e. $|d_{i-1}|\leq d_{max}$. Furthermore, as it depends on $u_{i-1}$, we suppose that it varies with $u_{i-1}$ and that it is a vanishing perturbation when $u_{i-1}$ is constant. 

% The model can be extended with the inclusion of a disturbance term, representing modeling errors, the missing modeling of the delays in (\ref{eq:generic_longitudinal_dynamics}) or external disturbances acting on the vehicle (see \cite{Besselink2017TAC}).

To derive the dynamics of the first vehicle $i=0$ of the platoon, in the same form of (\ref{eq:car_following_dynamics_i}), a non-autonomous non-communicating  virtual leader $i=-1$ is considered to precede the set of vehicles, with dynamical model
\begin{equation}\label{eq:virtual_leader}
    \dot{x}_{-1} = \left[
    \begin{array}{c}
         \dot{p}_{-1}  \\
         \dot{v}_{-1} 
    \end{array}
    \right] = \left[
    \begin{array}{c}
         v_{-1}  \\
         u_{-1}
    \end{array}
    \right].
\end{equation}
Then, the car-following dynamics with respect to vehicle $i=0$ can be described by:
\begin{equation}\label{eq:car_following_dynamics_0}
   \dot{\chi}_0 = \left[
    \begin{array}{c}
      \Delta \dot{p}_0    \\
      \Delta \dot{v}_0
    \end{array}
    \right] =  \left[
    \begin{array}{c}
         \Delta {v}_0 \\
         u_0-u_{-1}
    \end{array}
   \right].
\end{equation}
It follows that the dynamics in (\ref{eq:car_following_dynamics_i_short}) is valid $\forall \  i\in\mathcal{I}_N^0$. Since we consider $i=-1$ to represent a virtual vehicle, $p_{-1}(t) = \int_0^t v_{-1}(\tau)d\tau, \ t\geq 0,$ is a dummy state. Moreover, we consider $\Delta p_0(t) = -\Delta\Bar{p}, \ \forall t\geq 0$, where $\Delta\Bar{p} > 0$ is the constant desired inter-vehicular distance. \textcolor{black}{Since $p_i<p_{i-1}$, the desired distance has to be negative}.

In accordance to \cite{Ploeg2014TAC}, \cite{ZhengBorrelli2017TCST}, \cite{Zheng2014ITSC}, we consider the  widely accepted hypothesis of constant speed for the virtual leader $i=-1$, that precedes the entire cluster. Then, we have
\begin{equation}
    p_{-1}(t) = \Bar{v}\cdot t, \ v_{-1}(t) = \Bar{v}, \ u_{-1}(t) = 0, \ \forall t\geq 0
\end{equation}
where $\Bar{v}>0$ is a constant speed. 
%
% According to the constant speed assumption and (\ref{eq:virtual_leader}), it is simple to define the equilibrium point of the virtual leader as
% %
% \begin{equation}\label{eq:equilibrium_virtual_leader}
%     \Bar{x}_{-1}=[\Bar{v}\cdot t \:\: \Bar{v}]^T
% \end{equation}
% %
% for $u_{-1}=0$.
%
For vehicular platoons, the constant speed assumption defines the equilibrium point for all the vehicles in the cluster. Consequently, when all $i\in\mathcal{I}_N^0$ have equal speed and are at the same desired distance $\Delta\Bar{p}$, the equilibrium point for the $i$-th system of dynamics (\ref{eq:car_following_dynamics_i_short}) is
\begin{equation}\label{eq:chi_equilibrium}
   \chi_{e,i} = \Bar{\chi} = [\ -\Delta\Bar{p} \ \ \ 0 \ ]^T.
\end{equation}
Let $ \chi$ be the lumped state of the entire vehicle platoon:
\begin{equation}\label{eq:platoon_state}
    \chi=[\chi_0^T \ \chi_1^T \ ... \ \chi_N^T]^T.
\end{equation}
Then, for $u_{-1}=0$ it follows that 
\begin{equation}\label{eq:platoon_equilibrium}
     \chi_e = [ \Bar{\chi}^T \ \Bar{\chi}^T \ ... \ \Bar{\chi}^T]^T.
\end{equation}

%%%%%%%%%%%%%%%%%%%%%%%%%%%%%%%%%%%%%%%%%%%%%%%%%%%%%%%%%%%%%%%%%%%%%%%%%%%%%
\subsection{String Stability definitions}\label{sct:string_stability}

\begin{color}{black}

Let the model in (\ref{eq:car_following_dynamics_i}) describe a platoon, and its equilibrium be (\ref{eq:chi_equilibrium}). The control input $u_i$ is generated by the following dynamic controller
\begin{equation}\label{eq:dynamic_control_law}
    \begin{cases}
        \dot{\rho}_i = \omega_i( \rho_i , \chi ) \\
        u_i = h_i( \rho_i , \chi , \chi_{e,i} , u_{i-1} )
    \end{cases}
\end{equation}
where $\rho_i \in \mathbb{R}^r$, $r \geq 1$, is the state vector with dimension $r$ of the dynamic controller;  $\omega_i : \mathbb{R}^r \times \mathbb{R}^{2N}  \rightarrow \mathbb{R}^r$ is the vector field describing the evolution of the controller state; $h_i : \mathbb{R}^r \times \mathbb{R}^{2N} \times \mathbb{R}^2 \times \mathbb{R} \rightarrow \mathbb{R}$ is the output function of the system that corresponds to the control input $u_i$. The inputs of the dynamical system (\ref{eq:dynamic_control_law}) are $\chi, \ \chi_{e,i}, \ u_{i-1}$. The resulting closed loop system is denoted in the sequel by $P_{cl}$, where the state vector of the $i-$th vehicle and the corresponding equilibrium point are
\begin{equation}\label{eq:chi_extendedGENERAL}
    \Hat{\chi}_i = [ \: \chi_i^T \:\: \rho_i^T \: ]^T, \:\: \Hat{\chi}_{e,i} = [ \: \Bar{\chi}^T \:\: 0_{r}^T \: ]^T, \:\: \forall \: i \in\mathcal{I}_N^0,
\end{equation}
with $0_r\in\mathbb{R}^r$ a null column vector. 

We now recall the notions of String Stability and Asymptotic String Stability from \cite{Swaroop1996StringStability} and \cite{Besselink2017TAC}.
\begin{definition}{(String Stability)}\label{def:TSS_string_stability}
    The equilibrium $\Hat{\chi}_{e,i}, \: i\in\mathcal{I}_N^0$, of $P_{cl}$ is said to be String Stable if, for any $\epsilon>0$, there exists $\delta>0$ such that, for all $N\in\mathbb{N}$,
    \begin{equation}
        \max_{i\in\mathcal{I}_N^0}|\Hat{\chi}_i(0)-\Hat{\chi}_{e,i}|<\delta \Rightarrow \max_{i\in\mathcal{I}_N^0}|\Hat{\chi}_i(t)-\Hat{\chi}_{e,i}|<\epsilon, \ \forall \ t\geq 0.
    \end{equation}
\end{definition}
\begin{definition}{(Asymptotic String Stability)}\label{def:ATSS_string_stability}
    The equilibrium $\chi_{e,i}=\Bar{\chi}, \ i\in\mathcal{I}_N^0$, of $P_{cl}$ is said to be Asymptotically String Stable if it is String Stable and, for all $N\in\mathbb{N}$,
    \begin{equation}
        \lim_{t\rightarrow\infty}|\Hat{\chi}_i(t)-\Hat{\chi}_{e,i}|=0, \quad \forall \ i\in\mathcal{I}_N^0.
    \end{equation}
\end{definition}
\end{color}

\section{Control tools}\label{sct:control_tools}
The goal of this paper is to design a controller (\ref{eq:dynamic_control_law}) that adopts mesoscopic quantities and ensures asymptotic string stability of $P_{cl}$. To this purpose, a proper spacing policy  and a function describing macroscopic information are introduced. %According to Definitions (\ref{def:TSS_string_stability}) and (\ref{def:ATSS_string_stability}), the whole platoon is consequently shown to be string stable.

\subsection{Spacing policy}\label{sct:policies}
Several spacing policies have been introduced in the literature (see \cite{Swaroop1999DynSystCont}, \cite{Hedrick2004strings}). % A constant spacing strategy (see \cite{Swaroop1999DynSystCont}) consists in tracking a constant desired inter-vehicular distance $\Delta p^r = -\Delta\Bar{p}$, where $\Delta \Bar{p}>0$ is a constant value. Usually, the main drawback of such policy is that it does not guarantee String Stability if the control input of vehicle $i$ exploits only measurements with respect to its preceding vehicle $i-1$ (see \cite{Hedrick2004strings}). To stability purposes, additional information about other preceding vehicles must be shared.
%
%For example, the authors in \cite{Swaroop1996StringStability} develop a control law based on both the information of vehicle $i-1$ and vehicle $i=0$. The same methodology is used to develop the controller in \cite{Besselink2017TAC}, where a time-delayed spacing policy is proposed  to follow a constant inter-vehicular distance. As also discussed in \cite{ZhengBorrelli2017TCST}, assuming that the first vehicle of the platoon is allowed to send information to all the other vehicles is not practical. 
%
% To solve this issue, a macroscopic information based function is introduced in this paper.
%
%
%\subsubsection{Variable time spacing policy} 
We adopt a variable time spacing policy, which consists in tracking a variable inter-vehicular desired distance and allows for string stability and a low inter-vehicular spacing at steady-state (see \cite{Yanakiev1995CDC} and \cite{Zhang2005CDC}).
%
% This strategy relies in the family of nonlinear spacing policies. They have the advantage to ensure both String Stability and a low inter-vehicular spacing at regime conditions (see \cite{Yanakiev1995CDC} and \cite{Zhang2005CDC}). 
%
We define a \textit{mesoscopic} time varying trajectory for the distance policy $\Delta p^r_i$ of the $i$-th vehicle with respect to its leader $i-1$:
\begin{equation}\label{eq:mesoscopic_spacing_policy}
    \Delta p_{i}^r(t) = -\Delta\Bar{p} - \rho^M_i(t), \ t\geq 0
\end{equation}
where $\Delta\Bar{p}>0$ is the desired constant inter-vehicular distance and $\rho^M_i(t)$ is a function describing macroscopic information. Our goal is to show that, by using the macroscopic information, transient harmonization when traffic conditions vary is obtained while maintaining the platoon equilibrium in (\ref{eq:platoon_equilibrium}) in steady-state.
%
% \begin{color}{red}
%     The variable $\rho_i$ describes the information of the macroscopic variables and it is bounded as follows:
    
%     % To bound the contribution of the macroscopic information for the desired spacing policy with respect to the microscopic one, ad hoc bounds can be set for $\rho_i$:
%     %
%     \begin{equation}
%         \rho_i\in\left[ -\Delta\Bar{p}/2 , \Delta\Bar{p}/2 \right].
%     \end{equation}
%     %
% \end{color}
%

% We remark that the choice of $\rho_i$ does not change the platoon equilibrium, as the variance quantities converges to zero at steady-state. The goal is to slightly modify the reference trajectories when traffic conditions change for a better transient harmonization.%, and to converge to the static equilibrium when the platoon system is at regime.

%%%%%%%%%%%%%%%%%%%%%%%%%%%%%%%%%%%%%%%%%%%%%%%%%%%%%%%%%%%%%%%%%%%%%%%%%%%%%
%%%%%%%%%%%%%%%%%%%%%%%%%%%%%%%%%%%%%%%%%%%%%%%%%%%%%%%%%%%%%%%%%%%%%%%%%%%%%
%%%%%%%%%%%%%%%%%%%%%%%%%%%%%%%%%%%%%%%%%%%%%%%%%%%%%%%%%%%%%%%%%%%%%%%%%%%%%

\subsection{Macroscopic information}\label{sct:macroscopic_info}
Here we define proper macroscopic functions taking into account microscopic distance and speed variance, similarly to \cite{Iovine2015NAHS}. 
% We exploit the connection of the macroscopic traffic density and the microscopic speed variance of the vehicles in the platoon \cite{treiber_2013_book}. 
%From the point of view of traffic flow, information about inter-vehicular distances can be related to the traffic density; in fact, if we express the traffic density as the number of vehicles within a road segment of length $L_r>0$, then greater and lower densities are related to smaller and larger distances respectively.
%Target of this paper is to exploit macroscopic information about the traffic flow in the design of the control input $u_i(t)$ and in the definition of the spacing policy for each microscopic vehicle, in order to obtain a \textit{mesoscopic} control action for the platoon. As microscopic speed variance is related to the macroscopic density (see \cite{treiber_2013_book}), proper mesoscopic functions taking into account distance and speed variance will be used.
%
Given the generic vehicle $i\in\mathcal{I}_N$, let $\mu_{\Delta p,i}$ and $\sigma^2_{\Delta p,i}$ be the inter-vehicular distance mean and variance computed from vehicle $0$ to vehicle $i$, respectively:
\begin{equation}\label{eq:meanvar_distance}
    \mu_{\Delta p,i} = \frac{1}{i+1}\sum_{j = 0}^{i}\Delta p_j , \quad
    \sigma_{\Delta p,i}^2 = \frac{1}{i+1}\sum_{j = 0}^{i}(\Delta p_j-\mu_{\Delta p_i})^2.
\end{equation}
Let $\mu_{\Delta v,i}$ and $\sigma^2_{\Delta v,i}$ be the velocity tracking error mean and variance computed from vehicle $0$ to vehicle $i$, respectively:
\begin{equation}\label{eq:meanvar_speederror}
    \mu_{\Delta v,i} = \frac{1}{i+1}\sum_{j = 0}^{i}\Delta v_j, \quad
    \sigma_{\Delta v,i}^2 = \frac{1}{i+1}\sum_{j = 0}^{i}(\Delta v_j-\mu_{\Delta v_i})^2 .
\end{equation}
Let $\psi_{\Delta p}^i : \mathbb{R}\times\mathbb{R}^+ \rightarrow \mathbb{R}$ be the distance macroscopic function
%, where $\mathcal{D}=[\varphi_{\min},\varphi_{\max}]$, $\varphi_{\min}<\varphi_{\max}$, $\varphi_{\min},\varphi_{\max}\in\mathbb{R}$. 
and $\psi_{\Delta v}^i : \mathbb{R}\times\mathbb{R}^+ \rightarrow \mathbb{R}$ the speed tracking error macroscopic function, defined as
%where $\mathcal{V}=[\nu_{\min},\nu_{\max}]$, $\nu_{\min},\nu_{\max}\in\mathbb{R}$. 
% Then, we propose the macroscopic functions
%
% Similarly to \cite{Iovine2015NAHS}, let us consider the term $sign(\Delta\Bar{p}+\mu_{\Delta p,i-1}(\tau))$, which describes information about the sign of the variance. \textcolor{blue}{We recall that the signs are chosen accordingly to the definition of $\Delta p_i = p_i-p_{i-1}<0$}. Indeed, if $|\mu_{\Delta p,i}| < \Delta\Bar{p}$, then the vehicles are closing in each other. On the contrary, if $|\mu_{\Delta p,i}| > \Delta\Bar{p}$, then the vehicles are going away from each other. Also, let us define the term $\sqrt{\sigma^2_{\Delta p,i-1}}$, which describes information about the dispersion of the vehicles with respect to the mean. Similarly, let us define $sign(\mu_{\Delta v,i-1}(\tau))$ and $\sqrt{\sigma^2_{\Delta v,i-1}}$ with respect to the speeds; where the goal to pursue is $\Delta v_i=0, \ \forall i\in\mathcal{I}_N^0$. We propose the following macroscopic functions
%
\begin{equation}\label{eq:psi_delta_p}
    \psi_{\Delta p}^i = \gamma_{\Delta p} sign(\Delta\Bar{p}+\mu_{\Delta p,i})\sqrt{\sigma^2_{\Delta p,i}},
\end{equation}
\begin{equation}\label{eq:psi_delta_v}
    \psi_{\Delta v}^i = \gamma_{\Delta v} sign(\mu_{\Delta v,i})\sqrt{\sigma^2_{\Delta v,i}},
\end{equation}
\begin{color}{black}
where $\gamma_{\Delta p},\gamma_{\Delta v} > 0$ are constant parameters, $\mu_{\Delta p,i}$, $\mu_{\Delta v,i}$, $\sigma^2_{\Delta p,i}$ and $\sigma^2_{\Delta v,i}$ are defined in (\ref{eq:meanvar_distance}) and (\ref{eq:meanvar_speederror}), and
\begin{equation}
    sign(y) = 
    \begin{cases}
        1, \:\:\: & y > 0 \\
        0, & y = 0 \\
        -1, & y < 0
    \end{cases}
\end{equation}
%
%and the terms $sign(\Delta\Bar{p}+\mu_{\Delta p,i})$ and $sign(\mu_{\Delta v,i})$ in (\ref{eq:psi_delta_p}) and (\ref{eq:psi_delta_v}) are used to characterise the variance sign.
%
\end{color} 
\noindent Functions (\ref{eq:psi_delta_p}) and (\ref{eq:psi_delta_v}) connect the macroscopic density function with the variance of the microscopic distance and speed difference. \textcolor{black}{These functions catch the distance of the system with respect to its equilibrium.}
Instead of considering the whole set of leader-follower situations, they allow for a complexity reduction of the considered interconnected system without reducing the level of available information.

\begin{color}{black}
We embed the macroscopic information given by (\ref{eq:psi_delta_p}) and (\ref{eq:psi_delta_v}) in the macroscopic function denoted by 
%
% \begin{equation}
    $\rho_i = [ \rho_{1,i} \ \ \rho_{2,i} ]^T$, 
% \end{equation}
%
the evolution of which is given by the controller dynamics (\ref{eq:dynamic_control_law}), where we choose a state dimension $r = 2$ and an asymptotically stable dynamics:
% In order to relate the controller state evolution with the macroscopic information we propose the following asymptotically stable dynamical system:
%
\end{color}
\begin{equation}\label{eq:rho_dynamic_system}
    \begin{cases}
        \dot{\rho}_{1,i} = -\lambda_1 \rho_{1,i} + \rho_{2,i} \\
        \dot{\rho}_{2,i} = -\lambda_2 \rho_{2,i} + a \psi_{\Delta p}^{i-1} + b \psi_{\Delta v}^{i-1} \\
        \rho_{1,i}(0) = \rho_{2,i}(0) = 0 \\
        % \color{red}\rho_i \in [ -\Delta\Bar{p}/2 , \Delta\Bar{p}/2 ]
    \end{cases}
\end{equation}
%
% The evolution over the time of the system above is
% %
% \begin{align}\label{eq:rho_time_evolution}
%     %
%     \rho_i(t) =  \int_0^t e^{-\lambda(t-\tau)}(a \psi_{\Delta p}^i(\tau) + b \psi_{\Delta v}^i(\tau))d\tau
%     %
% \end{align}
% %
where $a,b\geq 0$ are chosen parameters, and $\lambda_1,\lambda_2 > 0$.
% represents the forgetting factors allowing $\rho_i(t)$ to have memory only of the recent evolution of the system. 
\textcolor{black}{The superscript $i-1$ in $\psi^{i-1}_{\Delta p}$ and $\psi^{i-1}_{\Delta v}$ means that we consider the macroscopic information calculated up to the preceding vehicle. Since no macroscopic information is available to vehicle $0$ we define $\psi^{-1}_{\Delta p} = \psi^{-1}_{\Delta v} = 0$.} Different macroscopic functions can be proposed, as in \cite{Iovine2015NAHS}. \textcolor{black}{The macroscopic function $\rho_i^M$ in (\ref{eq:mesoscopic_spacing_policy}) is defined as a linear combination of the components of $\rho_i$.  
%The main characteristic of $\rho_i^M$ (\ref{eq:rho_dynamic_system}) 
Its role is to incorporate the whole macroscopic information of the platoon avoiding complexity calculation explosion by the control law due to the state explosion. }
%
% In this framework, $\rho_i$ provides information related to macroscopic quantities. 
%Due to the necessity of showing String Stability results in a microscopic framework, 
%
%\textcolor{blue}{Functions (\ref{eq:meanvar_distance}) and (\ref{eq:meanvar_speederror}) can be obtained either through V2V or V2I communications, both of them with pros and cons. For example, in the former case it could be used distributed estimation algorithms in order to relieve the transmission burden adding uncertainties in the calculated quantites, however in the latter case those functions can be easily calculated by the road infrastructure, and then communicated to the vehicles. Consequently, the function $\rho_i$ in (\ref{eq:rho_dynamic_system}) can be obtained without V2V communication.}
\textcolor{black}{If the functions (\ref{eq:meanvar_distance}) and (\ref{eq:meanvar_speederror}) were to be obtained through V2V communications, a demanding exchange of information could be necessary. However, those functions can be easily calculated by the road infrastructure, and then communicated to the vehicles via V2I communications. Consequently, the function $\rho_i$ in (\ref{eq:rho_dynamic_system}) can be obtained without V2V communication. } 
\section{Mesoscopic control law}\label{sct:mesoscopic_control}

In this section, the control law adopting mesoscopic quantities for a single car-following situation is introduced. Then, String Stability and Asymptotic String Stability as in Definitions \ref{def:TSS_string_stability} and \ref{def:ATSS_string_stability} are ensured when the control laws are implemented for each leader-follower situation along the platoon. %The first control law uses the constant spacing policy in (\ref{eq:constant_spacing_policy}) and the function describing macroscopic information in (\ref{eq:rho_time_evolution}). 
The control law implements the variable spacing policy in (\ref{eq:mesoscopic_spacing_policy}) while considering the function $\rho_i$ in (\ref{eq:rho_dynamic_system}). Each vehicle is modeled according to dynamics (\ref{eq:longitudinal_dynamics}) and each car-following situation according to $\chi_i$ in (\ref{eq:car_following_dynamics_i}) and (\ref{eq:car_following_dynamics_0}). 
\begin{color}{black}
    To  analyze  the  String  Stability  of  the  closed  loop system, we consider the extended state (\ref{eq:chi_extendedGENERAL}) that includes the dynamics in (\ref{eq:rho_dynamic_system}).
    %
    % \begin{equation}\label{eq:chi_extendedGENERAL}
    %     \hat{\chi}_i = \left[
    %     \begin{array}{ccc}
    %         \Delta p_i  \ \ \
    %         \Delta v_i  \ \ \
    %         \rho_{i-1}^T
    %     \end{array}
    %     \right]^T, \ \forall \ i\in\mathcal{I}_N^0,
    % \end{equation}
    %    
    Defining $g_{cl,0}(\hat{\chi}_{-1}) = 0$, we can describe the closed loop dynamics of (\ref{eq:chi_extendedGENERAL}) as 
    \begin{align}
        \dot{\hat{\chi}}_0 &= f_{cl}(\hat{\chi}_0), \ i=0, \label{eq:closed_loop_platoon_short_0}\\
        \dot{\hat{\chi}}_i &= f_{cl}(\hat{\chi}_i) + g_{cl,i}(\hat{\chi}_{i-1},\hat{\chi}_{i-2},...,\hat{\chi}_0), \ \forall \  i\in\mathcal{I}_N. \label{eq:closed_loop_platoon_short_i}
    \end{align}  
    where $f_{cl}:\mathbb{R}^4 \rightarrow \mathbb{R}^4$ is the vector field describing the evolution dynamics of each isolated subsystem, and $g_{cl,i}:\underbrace{\mathbb{R}^4\times\cdot\cdot\cdot\times\mathbb{R}^4}_{i \text{ times}} \rightarrow \mathbb{R}^4$ is the interconnection term.
\end{color}
We assume that the virtual leader $i=-1$ has a constant speed $v_{-1} = \Bar{v}>0, \ u_{-1} = 0$. 
% Then, dynamics (\ref{eq:car_following_dynamics_i}) can be written with the following arguments:
% %
% \begin{align}
%     %
%     & \dot{\chi}_0 = f(\chi_0 , u_0 , 0 ),  &i = 0 \\
%     %
%     & \dot{\chi}_i = f(\chi_i , u_i , u_{i-1} ), &i\in\mathcal{I}_N
%     %
% \end{align}
% %
The assumption of $\Delta p_0(t) = -\Delta\Bar{p}, \ \forall t\geq 0$ is considered.

\subsection{Control strategy for variable spacing policy}

We present a control law \textcolor{black}{obtained by backstepping, see \cite{B_khalil_2002}),} for implementing the variable spacing policy in (\ref{eq:mesoscopic_spacing_policy}) with $\rho^M_i = \rho_{1,i}$. The controller associated to the $i-$th vehicle, $\forall \ i\in\mathcal{I}_N^0$, is given by: % coupling dynamics in (\ref{eq:rho_dynamic_system}) with the following one:
% %
% \begin{align}\label{eq:control_input_2}
%     %
%     \nonumber u_i &= u_{i-1} + \Delta\dot{v}_i^r - K_{\Delta v}(\Delta v_i - \Delta v_i^r) - (\Delta p_i + \Delta\Bar{p}) \\
%     %
%     & \quad -(1+K_{\Delta p}K_{\Delta v})\rho_{i-1}
%     %
% \end{align}
% %
%
\begin{align}\label{eq:control_input_2}
    \nonumber u_i &= u_{i-1} - (\Delta p_i + \Delta\Bar{p} + \rho_{1,i}) - \lambda_1(\lambda_1\rho_{1,i} - \rho_{2,i}) \\
    \nonumber &\quad + \lambda_2\rho_{2,i} - K_{\Delta p}(\Delta v_i - \lambda_1\rho_{1,i} +\rho_{2,i}) \\
    \nonumber &\quad - a\psi^{i-1}_{\Delta p} - b\psi^{i-1}_{\Delta v} - K_{\Delta v}(\Delta v_i - \lambda_1\rho_{1,i} + \rho_{2,i} \\
    \nonumber &\quad +K_{\Delta p}(\Delta p_i + \Delta\Bar{p} + \rho_{1,i})) \\
    \nonumber &= u_{i-1} - (\Delta p_i - \Delta p_i^r) - K_{\Delta v}(\Delta v_i - \Delta v_i^r) \\
    \nonumber &\quad + (K_{\Delta p} - \lambda_1)(\lambda_1\rho_{1,i} - \rho_{2,i}) +\lambda_2\rho_{2,i} \\
    &\quad -K_{\Delta p}\Delta v_i - a\psi^{i-1}_{\Delta p} - b\psi^{i-1}_{\Delta v},
\end{align}
with
\begin{align}\label{eq:delta_v_ref_i}
    \Delta v_i^r = \lambda_1\rho_{1,i}-\rho_{2,i}-K_{\Delta p}(\Delta p_i-\Delta p_i^r)
\end{align}
and given constant gains $K_{\Delta p},K_{\Delta v}>0$ equal for each $i\in\mathcal{I}_N^0$, $\Delta p_i^r = -\Delta\Bar{p}-\rho_{1,i}$, 
%as defined in (\ref{eq:mesoscopic_spacing_policy}),
%$\Delta v_i^r = \lambda_1\rho_{1,i-1}-\rho_{2,i-1}-K_{\Delta p}(\Delta p_i-\Delta p_i^r)$
%
% \textcolor{black}{
% \begin{align}
%     %
%     \Delta v_i^r &= -K_{\Delta p}(\Delta p + \Delta \Bar{p} + \rho_{i-1})\label{eq:delta_v_ref_1}
%     %
% \end{align}}
%
and $\rho_{1,i},\rho_{2,i}$ as defined in (\ref{eq:rho_dynamic_system}). 
%and (\ref{eq:rho_time_evolution}).

Note that the controller in (\ref{eq:control_input_2}) considers both microscopic and macroscopic information, leading to a \textit{mesoscopic} framework. 
%
%The control law in  (\ref{eq:control_input_2}) has the same capability of the one in (\ref{eq:control_input_1})  to exploit \textit{mesoscopic} information. 
To analyze the String Stability of the closed loop system, we consider the extended leader-follower state vector $\hat{\chi}_i$ and its corresponding equilibrium point $\Hat{\chi}_{e,i}$ in (\ref{eq:chi_extendedGENERAL}). %Its equilibrium point is 
%
% \begin{equation}\label{eq:chi_extended_equilibrium_2}
%     \hat{\chi}_{e,i} = [ \ -\Delta\Bar{p} \:\: 0 \:\: 0 \:\: 0 \ ]^T, \:\: \forall i\in\mathcal{I}_N^0.
% \end{equation}
%
%
% \begin{align}\label{eq:delta_v_ref_2}
%     \Delta v_i^r &= -K_{\Delta p}(\Delta p_i + \Delta\Bar{p}), \\
%     \Delta\dot{v}_i^r &= -K_{\Delta p}\Delta v_i.\label{eq:delta_v_ref_22}
% \end{align}
%
%$K_{\Delta p},K_{\Delta v}>0$ and $\rho_i$ as defined in (\ref{eq:rho_time_evolution}). 
%
The closed loop dynamics for each $i\in\mathcal{I}_N^0$ results to be:
%
% \begin{equation}\label{eq:car_following_dynamics_0_closed_loop}
%     \dot{\hat{\chi}}_0 = \left[  
%         \begin{array}{c}
%             \Delta\dot{p}_0  \\
%             \Delta\dot{v}_0  \\
%             \dot{\rho}_{1,0} \\
%             \dot{\rho}_{2,0} 
%         \end{array}
%     \right] = \left[  
%         \begin{array}{c}
%             \Delta v_0  \\
%             -(K_{\Delta p}+K_{\Delta v})\Delta v_0 \\
%             0 \\
%             0
%         \end{array}
%     \right]
% \end{equation}

%For vehicle $i=0$, the closed loop dynamics resulting from (\ref{eq:car_following_dynamics_0}) are the equal to the ones in (\ref{eq:car_following_dynamics_0_closed_loop}). Instead, the closed loop dynamics for vehicle $i>0$ in (\ref{eq:car_following_dynamics_i}) is:
%
\begin{align}\label{eq:car_following_dynamics_i_closed_loop_2}
    \dot{\hat{\chi}}_i =  \left[  
        \begin{array}{c}
            \Delta\dot{p}_i  \\
            \Delta\dot{v}_i  \\
            \dot{\rho}_{1,i} \\
            \dot{\rho}_{2,i} 
        \end{array}
    \right] = \left[  
    \begin{array}{c}
        \Delta v_i  \\
        (*) \\
        -\lambda_1\rho_{1,i} + \rho_{2,i} \\
        -\lambda_2\rho_{2,i} + a\psi_{\Delta p}^{i-1} + b\psi_{\Delta v}^{i-1}
    \end{array}
    \right]
\end{align}
with 
\begin{align*}
    (*) &= - (\Delta p_i - \Delta p_i^r) - K_{\Delta v}(\Delta v_i - \Delta v_i^r) \\
    \nonumber &\quad + (K_{\Delta p} - \lambda_1)(\lambda_1\rho_{1,i} - \rho_{2,i}) +\lambda_2\rho_{2,i-1} \\
    &\quad -K_{\Delta p}\Delta v_i - a\psi^{i-1}_{\Delta p} - b\psi^{i-1}_{\Delta v}.
\end{align*}
\textcolor{black}{We remark that $\psi^{-1}_{\Delta p} = \psi^{-1}_{\Delta v} = 0$.} Then, we can rewrite the system in %(\ref{eq:car_following_dynamics_0_closed_loop}) and
(\ref{eq:car_following_dynamics_i_closed_loop_2}) as (\ref{eq:closed_loop_platoon_short_0}) and (\ref{eq:closed_loop_platoon_short_i}), 
where $g_{cl,i}(\hat{\chi}_{i-1},...,\hat{\chi}_0)$ is 
\begin{equation}\label{eq:car_following_dynamics_g_cl_1}
    g_{cl,i}(\hat{\chi}_{i-1},\hat{\chi}_{i-2},...,\hat{\chi}_0) = \left[ 
        \begin{array}{c}
            0  \\
            -(a\psi_{\Delta p}^{i-1} + b\psi_{\Delta v}^{i-1}) \\
            0  \\
            a\psi_{\Delta p}^{i-1} + b\psi_{\Delta v}^{i-1}
        \end{array}
    \right].
\end{equation}
%
%Since $\rho_{-1} = g_{\rho,-1}(\chi_{-1}) = 0$, 
\subsection{String Stability analysis}
 
% \begin{color}{red}
%     Furthermore, we let each agent of the interconnected system depend from a neighborhood $\mathcal{N}_i = {i-1,i-2,...,0}$ \cite{Besselink2018CSL}. The main difference is that our modeling describes the macroscopic term to this purpose.
% \end{color}
%
\noindent
Set $\Tilde{\chi}_i = \hat{\chi}_{i} - \hat{\chi}_{e,i}$. Then, the following result holds:
\begin{lemma}\label{lmm:closed_loop_ISS_control2}
    Consider the closed loop system described by %(\ref{eq:car_following_dynamics_0_closed_loop}) and
    (\ref{eq:car_following_dynamics_i_closed_loop_2}). Then, there exist functions $\beta$ of class $\mathcal{KL}$ and $\gamma$ of class $\mathcal{K}_\infty$ such that, if $ K_{\Delta p}, K_{\Delta v}, \lambda_1, \lambda_2 > 0$
    %
    % \begin{equation}\label{eq:kappalambda_inequality}
    %     \min\{ K_{\Delta v} , K_{\Delta p} + K_{\Delta v}K^2_{\Delta p} \} - \sqrt{3}( 1 + K^2_{\Delta p} )\lambda> 0
    % \end{equation}
    %
    then,
    \begin{equation}\label{eq:closed_loop_platoon_ISS_ineq2}
        |\Tilde{\chi}_i(t)| \leq \beta( |\Tilde{\chi}_i(0)| , t ) + \gamma\left( \max_{j=0,...,i-1} |\Tilde{\chi}_j(\cdot)|_\infty^{[0,t]} \right)
    \end{equation}
    $\forall t\geq 0$, and $\gamma(s) = \Tilde{\gamma}s,\ s\geq 0$, $\Tilde{\gamma}\in\mathbb{R}^+$. Moreover, there exist 
    % gains $K_{\Delta p},K_{\Delta v}$ in (\ref{eq:control_input_2}) and 
     $a$ and $b$ in (\ref{eq:rho_dynamic_system}) such that $\Tilde{\gamma} \in (0,1)$. 
\end{lemma}
\begin{proof}\color{black}
    See Appendix \ref{apx:closed_loop_ISS_control2_proof} in \cite{Mirabilio2020CDC_Extended}.
\end{proof}
\noindent
On the basis of Lemma \ref{lmm:closed_loop_ISS_control2}, Asymptotic String Stability of the platoon can be obtained by an appropriate choice of the parameters in (\ref{eq:car_following_dynamics_i_closed_loop_2}), 
% using an appropriately chosen function describing macroscopic information, 
as shown in the following:

\begin{theorem}\label{thm:string_stability_control2}
    % There exists a choice of $K_{\Delta p}, \ K_{\Delta v}, \ \lambda > 0$ such that $\Tilde{\gamma}\in(0,1)$ in (\ref{eq:gammatilde}). Consequently, the closed loop system in (\ref{eq:car_following_dynamics_0_closed_loop}) and (\ref{eq:car_following_dynamics_i_closed_loop_2})  is String Stable.
    % Closed loop system described by (\ref{eq:car_following_dynamics_0_closed_loop}) and (\ref{eq:car_following_dynamics_i_closed_loop_2}), satisfying the condition in Lemma \ref{lmm:closed_loop_ISS_control2} with $\Tilde{\gamma}\in(0,1)$, is String Stable.
    The closed loop system described by %(\ref{eq:car_following_dynamics_0_closed_loop}) and 
    (\ref{eq:car_following_dynamics_i_closed_loop_2}) where the parameters $K_{\Delta p}, K_{\Delta v}, \lambda_1, \lambda_2 > 0$ and parameters $a,b$ are such that $\tilde{\gamma}\in (0,1)$, is Asymptotically String Stable.
\end{theorem}
\begin{proof}\color{black}
    See Appendix \ref{apx:string_stability_control2_proof} in \cite{Mirabilio2020CDC_Extended}.
\end{proof}
\section{Simulations}\label{sct:simulations}

The introduced control strategy is simulated in Matlab\&Simulink. Based on the modeling in (\ref{eq:car_following_dynamics_i}), we consider a platoon of $N+1=11$ vehicles. 
The initial conditions for each vehicle are randomly generated in a neighborhood of the equilibrium point. It results $\mu_{\Delta p}\neq-\Delta\Bar{p},\mu_{\Delta v}\neq 0$ $\sigma^2_{\Delta p},\sigma^2_{\Delta v} \neq 0$. The reference distance is $\Delta\Bar{p} = 10m$ and the initial desired speed of the leading vehicle is $\Bar{v} = 14m/s$. Vehicle speed is  $0 < v_i\leq 36$ $[m/s]$ and the acceleration is bounded by $-4\leq u_i\leq 4$ $[m/s^2]$. The control parameters are introduced in Table \ref{table:parameters}, with a resulting $\Tilde{\gamma}=0.5$. To better stress the advantages of the proposed controller, we analyze the behavior of the system when a disturbance acts on the acceleration of vehicle $i=0$, and it is not communicated to vehicle $i=1$. 

\noindent The simulation time is 1 minute. We split it into three phases:
%
% \begin{center}
% \begin{table}
%     \caption{The control parameters.}
%     \begin{tabular}{ | c | c | c | c |c | c | }
%     \hline
%   Parameter  & Value & Parameter & Value & Parameter & Value  \\ \hline
%   $K_{\Delta p}$ & $1$ & $K_{\Delta v}$ & $2$  & $\Upsilon$ & $0.9$   \\ \hline
%   $\lambda_1,\lambda_2$ & $1.5$ & $a$ & $0.6$ & $b$ & $0.6$  \\  \hline
%   $\gamma_{\Delta p}$ & $0.5$ & $\gamma_{\Delta v}$ & $0.5$ & $\Tilde{\gamma}$ & $0.5$\\  \hline
% %  $\gamma_{\Delta v}$ & $0.2$ & $\Upsilon$ & $0.9$  \\  \hline
%     \end{tabular}\label{table:parameters}
%     \end{table}
% \end{center}
\begin{center}
\begin{table}
    \caption{The control parameters.}
    \begin{tabular}{ | c | c | c | c |c | c | }
    \hline
   Parameter  & Value & Parameter & Value & Parameter & Value  \\ \hline
  $K_{\Delta p}$ & $1$ & $K_{\Delta v}$ & $2$  & $\Upsilon$ & $0.9$   \\ \hline
  $\lambda_1,\lambda_2$ & $1.5$ & $a$ & $0.2$ & $b$ & $1$  \\  \hline
  $\gamma_{\Delta p}$ & $0.5$ & $\gamma_{\Delta v}$ & $0.5$ & $\Tilde{\gamma}$ & $0.5$\\  \hline
%  $\gamma_{\Delta v}$ & $0.2$ & $\Upsilon$ & $0.9$  \\  \hline
    \end{tabular}\label{table:parameters}
    \end{table}
\end{center}
\vspace{-10pt}
\begin{enumerate}

    \item From $t=t_0=0s$ to $t=t_1=10s$: the vehicles start with initial conditions that are different from the desired speed and the desired distance. No disturbance is acting on the leader vehicle, and its desired speed is the initial one, i.e. $\Bar{v} = 14m/s$.
    
    \item From $t=t_1=10s$ to $t=t_2=30s$: a disturbance acts on the acceleration of the first vehicle $i=0$. At $t_1$ a positive pulse of amplitude $4m/s^2$ and length $5s$ is considered, while a similar pulse with negative amplitude is considered at $t=15s$. The control input of $i=0$ being saturated, $i=0$ succeeds to properly counteract to it but it is not able to operate the needed corrective action to return to the desired speed. Since the disturbance is an external input, it is not communicated to the follower and can propagate along the platoon.%through the string amplifying its action. In Fig.(\ref{fig:speeds_1}) and (\ref{fig:accelerations_1}) we can see that the controller is able to attenuate this disturbance, showing String Stability performance. Moreover, in the zoomed section we can see that the macroscopic information let the vehicles to anticipate their action: the violet dashed line corresponding to the last vehicle is slightly anticipated with respect to the response of the first vehicels.
    
    \item From $t=t_2=30s$ to $t=t_3=60s$: the leader tracks a variable speed reference. From $t=30s$ to $t=45s$ the desired speed is $\Bar{v}=30m/s$, while from $t=45s$ to $t=60s$ it is $\Bar{v}=20m/s$. %Vehicles success in track the variable speed and to remain at the desired distance. Contrariwise to the disturbance case, the vehicles success to follow each other without oscillations because the information spread through $u_i$ is the real applied.  
    
\end{enumerate}

Figures \ref{fig:ditances_2}, \ref{fig:speeds_2} and \ref{fig:accelerations_2} show, respectively, the inter-vehicular distance, speed and acceleration profiles for each vehicle of the platoon when the control input in (\ref{eq:control_input_2}) is implemented. In the first phase, the vehicles are shown to quickly converge to the desired speed and the desired distance. 
%
%Figures \ref{fig:ditances_2}, \ref{fig:speeds_2} and \ref{fig:accelerations_2} show the platoon dynamical evolution when the control input (\ref{eq:control_input_2}) is implemented. Similarly to the case of constant spacing policy,  the vehicles are shown to converge to the desired speed and the desired distance in short time in this first phase of the simulation. 

In the second phase, the controller of $i=1$ does not know \textcolor{black}{the correct acceleration value of vehicle $i=0$}. 
Also, the macroscopic variable is not available to it: for these reasons, it does not succeed to perfectly track the desired distance neither in case of positive disturbance between $t=10s$ to  $t=15s$ nor in case of negative one between $t=15s$ to  $t=20s$ (see Figure \ref{fig:ditances_2}). However, it converges to the same speed of $i=0$ after a small transient of three seconds in both cases (see Figure \ref{fig:speeds_2}). Finally, at $t=20s$ the disturbance is not active anymore and the leader can restore its desired speed. Also, $i=1$ receives correct information about its leader acceleration and is able to return to the ideal distance. The dynamical evolution of the remaining vehicles in the platoon during the generated transients after $t=10s$, $t=15s$ and $t=20s$ catches the contribution of the macroscopic information. To this purpose, let us consider the speed dynamics of the last vehicle in Figure \ref{fig:speeds_2}. It is possible to remark an anticipatory behaviour due to the macroscopic information resulting in a higher speed between $t=10s$ and $t=12s$ with respect to the leading vehicles%the vehicles along the platoon scale to intensify their accelerations and speeds. Moreover, the same anticipatory behaviour is shown when the leading vehicles are converging to the same speed. Here, 
\textcolor{black}{. Then, } the decreasing of vehicles' speed along the platoon scales with respect to their position, resulting more stressed in the last vehicles (see between $t=12s$ and $t=14s$). The same anticipatory behaviour is shown in Figure \ref{fig:accelerations_2} with respect to the accelerations of the leading vehicles. The acceleration profiles better show how the vehicles along the platoon scale to intensify their accelerations and speeds, both for increasing and decreasing speed phases. An anticipatory behaviour is shown, both when the leading vehicles are accelerating and converging to the same speed. 
The same applies for transients taking place after $t=15s$ and $t=20s$, which are generated by the fast reaction of the leading vehicle to the disturbance.  %The same situations of the previous simulation apply to the second phase. Indeed, a
An anticipatory harmonizing acceleration for each vehicle scales along the platoon with respect to their knowledge of the macroscopic quantities, as shown in Figure \ref{fig:speeds_2} and \ref{fig:accelerations_2}. %The main difference in this case relies on a less important anticipatory reaction when the platoon counteracts to the external disturbance. This is due to the double utilization of the macroscopic information. Indeed, since now also the spacing policy depends on it, the whole platoon acts more uniformly, and shows more harmonization in the reaction phase. 

In the third phase, since there is no unknown perturbation acting on the platoon, the vehicles succeed to track the variable speed profile and to remain at the desired distance. No oscillations are shown by the proposed control law, even if the desired speed profile has high steps.%  Contrariwise to the disturbance case, the vehicles success to follow each other without oscillations because the information spread through $u_i$ is the real applied.  

\begin{figure}[!]
    \centering
    \includegraphics[width = 1\columnwidth]{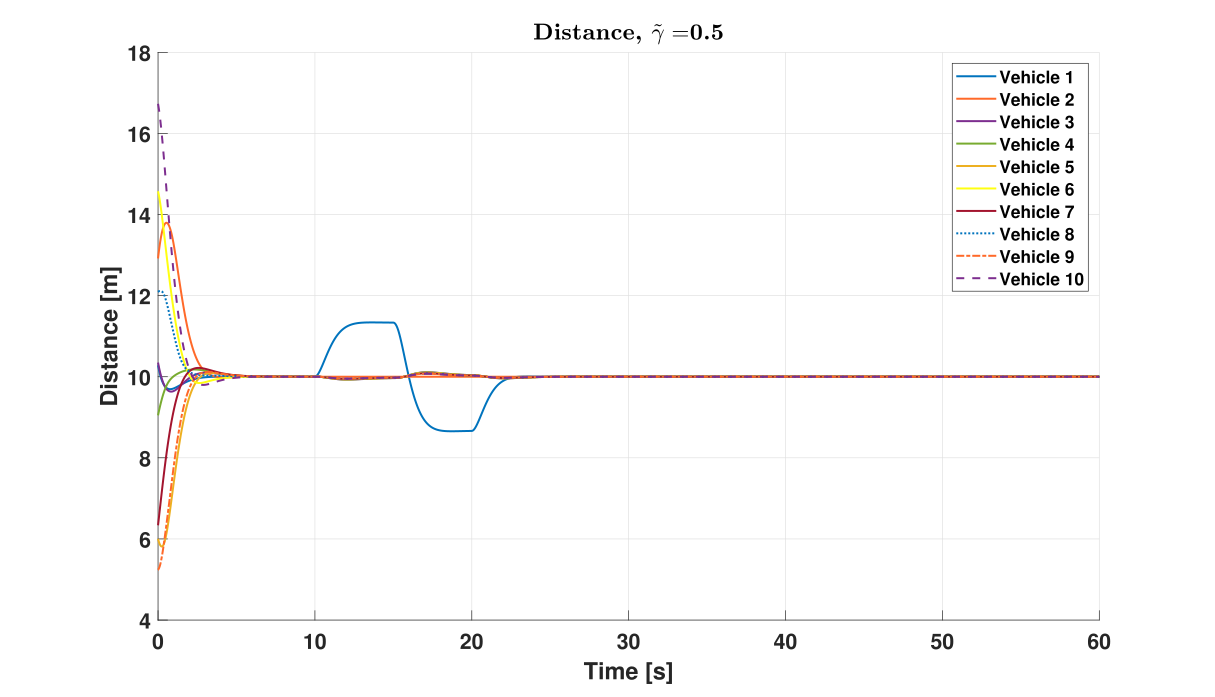}
    \caption{Control strategy for variable spacing policy: Distances.}
    \label{fig:ditances_2}
    \vspace{-15pt}
\end{figure}

% \begin{figure}[!]
%     \centering
%     \includegraphics[width = 1\columnwidth]{figures/distance_ctrl3_bistriplo2.png}
%     \caption{Control strategy for variable spacing policy: Distances.}
%     % \label{fig:ditances_2}
%     \vspace{-10pt}
% \end{figure}

\begin{figure}[!]
    \centering
    \includegraphics[width = 1\columnwidth]{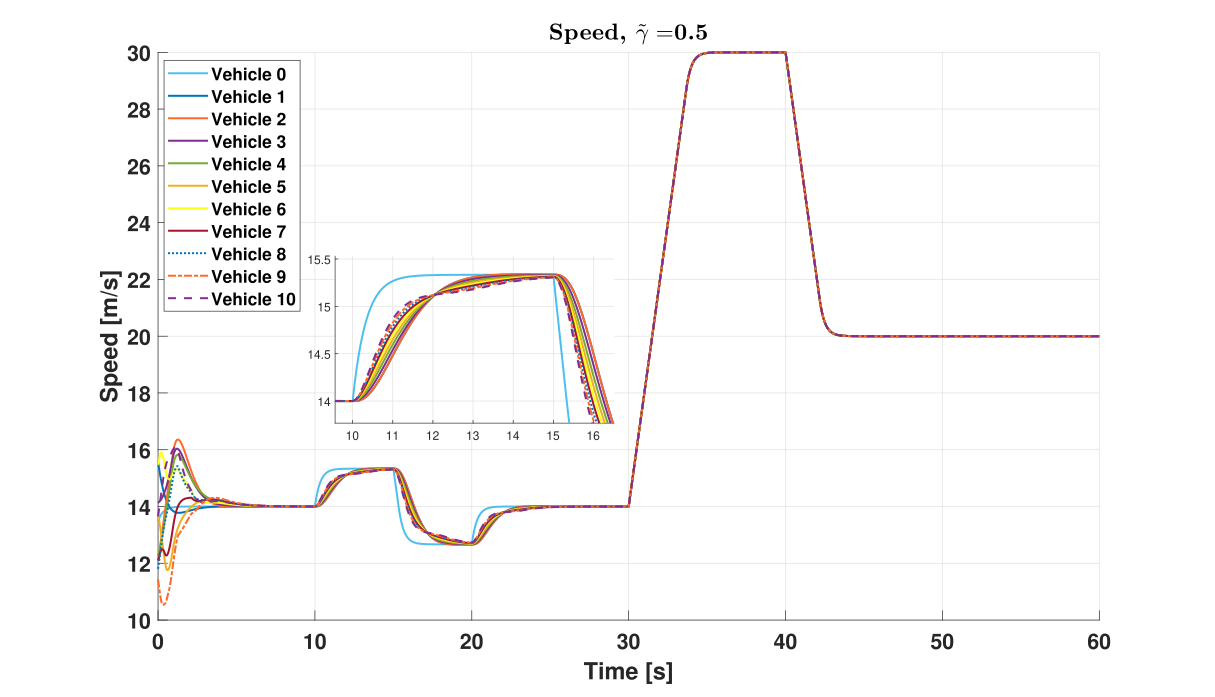}
    \vspace{-15pt}
    \caption{Control strategy for variable spacing policy: Speeds.}
    \label{fig:speeds_2}
    \vspace{-10pt}
\end{figure}

%\begin{figure}[!]
%    \centering
%    \includegraphics[width = 1\columnwidth]{figures/speed_ctrl3_bistriplo2.png}
%    \caption{Control strategy for variable spacing policy: Speeds.}
    % \label{fig:speeds_2}
%    \vspace{-10pt}
%\end{figure}

\begin{figure}[!]
    \centering
    \includegraphics[width = 1\columnwidth]{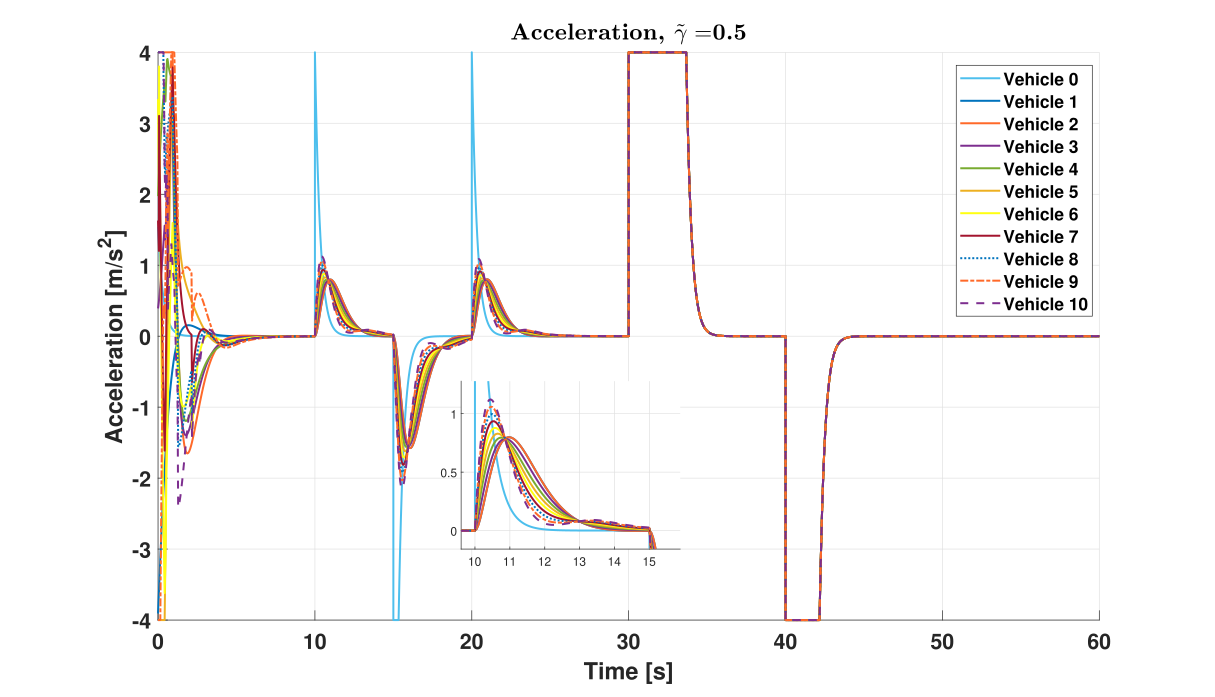}
    \vspace{-10pt}
    \caption{Control strategy for variable spacing policy: Accelerations.}
    \label{fig:accelerations_2}
    \vspace{-15pt}
\end{figure}

%\begin{figure}[!]
%    \centering
%    \includegraphics[width = 1\columnwidth]{figures/acceleration_ctrl3_bistriplo2.png}
%    \caption{Control strategy for variable spacing policy: Accelerations.}
%    % \label{fig:accelerations_2}
%    \vspace{-10pt}
%\end{figure}

The proposed control laws in (\ref{eq:control_input_2}) exploit the information resulting from the macroscopic variable and safely control a platoon of vehicles. The control inputs provide transient harmonization on the whole traffic throughput while ensuring Asymptotic String Stability properties. The dynamical evolution results in a reduction of the oscillations propagation along the platoon, both in nominal case and in the presence of an active external disturbance. The utilization of the macroscopic information results to be a powerful tool. 
%Future works need to reduce the amount of information exchanged to propagate macroscopic quantities, focusing on V2I frameworks with aggregate and recursive variables.

%Also in this case the vehicles success in following each other at the desired distance and with same speed. As before, the system show String Stability performance, but the controller results in a slightly worst behavior; it does not show the anticipatory action, but is able to attenuate the external disturbance acting on the first vehicle.

%%%%%%%%%%%%%%%%%%%%%%%%%%%%%%%%%%%%%%%%%%%%%%%%%%%%%%%%%%%%%%%%%%%%%%%%%%%%%
%%%%%%%%%%%%%%%%%%%%%%%%%%%%%%%%%%%%%%%%%%%%%%%%%%%%%%%%%%%%%%%%%%%%%%%%%%%%%
\section{Conclusions}\label{sct:conclusion}
This paper introduces macroscopic variables for ensuring the String Stability of a platoon of CACC autonomous vehicles. As the variance of microscopic quantities is related to the macroscopic density, the proposed stability analysis opens to the possibility of properly controlling a platoon by propagating only macroscopic density information. %In this paper, this is done in the microscopic framework via V2V communication, but the possibility to use V2I communication is envisaged. %instead that several microscopic variables as speed, position and acceleration via V2V communication. % , that can be considered even . 
%Two control laws have been proposed, and a rigorous analysis has been performed. The variance being related to the macroscopic %, having the necessity to ensure String Stability in a platoon of CACC as main goal. 
A control law based on information obtained by V2V communication has been proposed. The improvements resulting from taking into account macroscopic information are shown by simulation results. The proposed mesoscopic control law produces an anticipatory behaviour, which provides a better transient harmonization. Future work will focus on 
%the possibility to consider only V2I communications for macroscopic quantities sharing, and 
extending the proposed framework in a mixed traffic situation with non-communicating vehicles. \textcolor{black}{Also, due to the satisfactory results,  the next objective is to extend our approach to more complex models including non-idealities such as actuation and communication delays.}

% \begin{figure}
%     \centering
%     \includegraphics[width = 1\columnwidth]{figures/phase_ctrl3_bisdoppio.png}
%     \caption{Control strategy 1: Phase diagram.}
%     \label{fig:phases_1}
% \end{figure}

% \begin{figure}
%     \centering
%     \includegraphics[width = 1\columnwidth]{figures/phase_ctrl3_bistriplo.png}
%     \caption{Control strategy 2: Phase diagram.}
%     \label{fig:phases_2}
% \end{figure}

% \bibliographystyle{IEEEtran} 
\bibliographystyle{ieeetr}
\bibliography{biblioTraffic_20200820.bib}

%%%%%%%%%%%%%%%%%%%%%%%%%%%%%%%%%%%%%%%%%%%%%%%%%%%%%%%%%%%%%%%%%%%%%%%%%%%%%
%%%%%%%%%%%%%%%%%%%%%%%%%%%%%%%%%%%%%%%%%%%%%%%%%%%%%%%%%%%%%%%%%%%%%%%%%%%%%
\clearpage
\appendix
\section{Appendices}
%\subsection{First appendix}

%%%%%%%%%%%%%%%%%%%%%%%%%%%%%%%%%%%%%%%%%%%%%%%%%%%%%%%%%%%%%%%%%%%%%%%%%%%%%
%%%%%%%%%%%%%%%%%%%%%%%%%%%%%%%%%%%%%%%%%%%%%%%%%%%%%%%%%%%%%%%%%%%%%%%%%%%%%
\subsection{Proof of Lemma \ref{lmm:closed_loop_ISS_control2}}\label{apx:closed_loop_ISS_control2_proof}

Let us consider the candidate Lyapunov function (see \cite{B_khalil_2002}) $W_i = W(\Tilde{\chi}_i)$ for the $i$-th dynamical system $\Tilde{\chi}_i$, for $i\in\mathcal{I}_N^0$:
\begin{align}\label{eq:LyapunovFunction_ith_2}
    \nonumber  W(\Tilde{\chi}_i) &= \frac{1}{2}(\Delta p_i - \Delta p_i^r)^2 + \frac{1}{2}(\Delta v_i - \Delta v_i^r)^2 \\
    &\quad + \frac{1}{2}\rho_{1,i}^2 + \frac{1}{2}\rho_{2,i}^2.
\end{align}
% \begin{align}\label{eq:LyapunovFunction_ith_2}
%     %
%     \nonumber & W(\Tilde{\chi}_i) = \frac{1}{2}(\Delta p_i + \Delta\Bar{p} + \rho_{1,i-1})^2 + \frac{1}{2}\rho_{1,i-1}^2 + \frac{1}{2}\rho_{2,i-1}^2 \\
%     %
%     \nonumber &\quad +\frac{1}{2}(\Delta v_i - \lambda_1\rho_{1,i-1} + \rho_{2,i-1} +\\
%     %
%     &\quad + K_{\Delta p}(\Delta p_i + \Delta\Bar{p} + \rho_{1,i-1}))^2.
%     %
% \end{align}
%
It can be proven that (\ref{eq:LyapunovFunction_ith_2}) is bounded by two quadratic functions as follows
\begin{equation}\label{eq:W_i_bounded}
    \underline{\alpha}|\Tilde{\chi}_i|^2 \leq W(\Tilde{\chi}_i) \leq \Bar{\alpha}|\Tilde{\chi}_i|^2.
\end{equation}
To derive the constants $\underline{\alpha} > 0$ and $\Bar{\alpha} > 0$ we replace the expressions of $\Delta p_i^r$ defined in (\ref{eq:mesoscopic_spacing_policy}) and of $\Delta v_i^r$ in (\ref{eq:delta_v_ref_i}) in the definition of $W_i$ in (\ref{eq:LyapunovFunction_ith_2}). Then, 
\begin{align}
    \nonumber & W(\Tilde{\chi}_i) = \frac{1}{2}(\Delta p_i + \Delta\Bar{p} + \rho_{1,i})^2 \\
    \nonumber & \ + \frac{1}{2}(\Delta v_i - \lambda_1\rho_{1,i} + \rho_{2,i} + K_{\Delta p}(\Delta p_i + \Delta\Bar{p} + \rho_{1,i}))^2 \\
    \nonumber & \ + \frac{1}{2}\rho^2_{1,i} + \frac{1}{2}\rho^2_{2,i} \\
    & \ = \frac{1}{2} \Tilde{\chi}_i^T \underbrace{\left[ 
    \begin{array}{cccc}
        1+K^2_{\Delta p} & p_1 & p_{2} & p_1 \\
        0 & 1 & p_3 & 2 \\
        0 & 0 & 2+(\lambda_1-K_{\Delta p})^2 & p_3 \\
        0 & 0 & 0 & 2
    \end{array}
    \right]}_{P_W} \Tilde{\chi}_i
\end{align}
where
\begin{align*}
    &p_1 = 2K_{\Delta p}, \ p_2 = 2(1+K^2_{\Delta p}-\lambda_1 K_{\Delta p}),
    \\ &p_3 = 2(K_{\Delta p}-\lambda_1).
\end{align*}
By defining $\lambda_{\min}(P),\lambda_{\max}(P)$ respectively the minimum and maximum eigenvalues of the generic matrix P, then,
\begin{equation}\label{eq:W_i_lambda_bounds}
    \frac{1}{2}\lambda_{\min}(P_W)|\Tilde{\chi}_i|^2 \leq W(\Tilde{\chi}_i) \leq \frac{1}{2}\lambda_{\max}(P_W)|\Tilde{\chi}_i|^2.
\end{equation}
Since $K_{\Delta p},\lambda_1,\lambda_2 > 0$, then 
\begin{equation}\label{eq:W_i_lambda_definitions}
    \lambda_{\min}(P_W) = 1, \ \lambda_{\max}(P_W) = \max\{ 1+K^2_{\Delta p} , 2+(\lambda_1-K_{\Delta p})^2 \}.
\end{equation}
From (\ref{eq:W_i_lambda_bounds}) and (\ref{eq:W_i_lambda_definitions}) it follows
\begin{equation}\label{eq:conditions_alpha_i}
    \underline{\alpha} = \frac{1}{2}, \quad \Bar{\alpha}  = \frac{1}{2}\max\{ 1+K_{\Delta p}^2 , 2+(\lambda_1-K_{\Delta p})^2 \}.
\end{equation} 
%
%\textcolor{black}{(see Appendix \ref{App_1} in \cite{Mirabilio2020CDC_Extended} for more details)}. %We denote shortly $W_i = W(\Tilde{\chi}_i)$ and $\dot{W}_i = \dot{W}(\Tilde{\chi}_i)$.
%
The time derivative of $W_i$ in (\ref{eq:LyapunovFunction_ith_2}) verifies:
\begin{align}\label{eq:dot_LyapunovFunction_ith_2_inequality1}
    \nonumber \dot{W}_i &= -K_{\Delta p}(\Delta p_i - \Delta p_i^r)^2 -K_{\Delta v}(\Delta v_i - \Delta v_i^r)^2  \\
    \nonumber &\quad -\lambda_1\rho^2_{1,i} - \lambda_2\rho^2_{2,i} + \rho_{1,i}\rho_{2,i} \\
    \nonumber &\quad  + \rho_{2,i}(a\psi^{i-1}_{\Delta p} + b\psi^{i-1}_{\Delta v}) \\
    \nonumber & = - \Tilde{\chi}_i^T \underbrace{\left[ 
    \begin{array}{cccc}
        q_1 & 2K_{\Delta p}K_{\Delta v} & q_2 & 2K_{\Delta p}K_{\Delta v} \\
        0 & K_{\Delta v} & q_3 & 2K_{\Delta v} \\
        0 & 0 & q_4 & q_5 \\
        0 & 0 & 0 & \lambda_2+K_{\Delta v}
    \end{array}
    \right]}_{Q_{W}}
    \Tilde{\chi}_i \\
    \nonumber & \ + |\Tilde{\chi}_i|(a|\psi^{i-1}_{\Delta p}| + b|\psi^{i-1}_{\Delta v}|) \\
    &\leq -{\alpha}|\Tilde{\chi}_i|^2 + |\Tilde{\chi}_i|(a|\psi^{i-1}_{\Delta p}| + b|\psi^{i-1}_{\Delta v}|)
\end{align}
where 
\begin{align}
    \alpha &= \lambda_{\min}(Q_W) = \min\{ q_1 , K_{\Delta v} , q_4 , \lambda_2+K_{\Delta v} \}, \label{eq:alpha} 
\end{align}
with
\begin{align*}
    & q_1 = K_{\Delta p}(1+K_{\Delta p}K_{\Delta v}), \\
    & q_2 = 2K_{\Delta p}(1+K_{\Delta v}(K_{\Delta p}-\lambda_1)), \\
    & q_3 = 2K_{\Delta v}(K_{\Delta p}-\lambda_1), \\
    & q_4 = K_{\Delta p}+\lambda_1+K_{\Delta v}(\lambda_1-K_{\Delta p})^2, \\
    & q_5 = 1-2K_{\Delta v}(K_{\Delta p}-\lambda_1).
\end{align*}
%
%$c_1 = K_{\Delta p}(1+K_{\Delta p}K_{\Delta v}), \ c_2 = K_{\Delta v}, \ c_3 = K_{\Delta p}+\lambda_1+K_{\Delta v}(\lambda_1-K_{\Delta p})^2, \ c_4 = \lambda_2+K_{\Delta v}$. \textcolor{black}{See Appendix \ref{App_2} in \cite{Mirabilio2020CDC_Extended} for more details on the calculations in (\ref{eq:dot_LyapunovFunction_ith_2_inequality1}).} 
% Let us define $\alpha = \Tilde{\alpha}-\Tilde{k}\lambda$, 
% We define $d = a\gamma_{\Delta p}+b\gamma_{\Delta v} > 0$ and $\Upsilon\in(0,1)$. 
We proceed to prove some inequalities with respect to functions $\psi_{\Delta p}^i$ and $\psi^i_{\Delta v}$ by exploiting the variance property given below. Let $l\in\{1,...,m\}$ and $y_l\in\mathbb{R}$. Then, the variance with respect to the set of values $y_l$ satisfies the property 
\begin{equation}\label{eq:variance_property}
    \sigma_y^2 \leq \frac{1}{4}(\max_l y_l - \min_l y_l)^2.
\end{equation}
Since we consider the dynamics in (\ref{eq:closed_loop_platoon_short_0}) and (\ref{eq:closed_loop_platoon_short_i}) with respect to $\Tilde{\chi}_i = \hat{\chi}_i-\hat{\chi}_{e,i}$, let us define $\Delta\Tilde{p}_i = \Delta p_i -  \Delta\Bar{p}$, then
\begin{align}
    \nonumber |\psi^i_{\Delta p}| &\leq \gamma_{\Delta p}\sqrt{\sigma^2_{\Delta p}} \\
    \nonumber &\leq \frac{1}{2}\gamma_{\Delta p}|\max_{j=0,...,i} \Delta\Tilde{p}_j - \min_{j=0,...,i} \Delta\Tilde{p}_j| \\
    \nonumber &\leq \gamma_{\Delta p}\max_{j=0,...,i}|\Delta\Tilde{p}_j| \\
    %
    % \nonumber &= \gamma_{\Delta p}\max_{j=0,...,i}\sqrt{ (\Delta\Tilde{p}_j)^2 + 0\cdot(\Delta v_j)^2 + 0\cdot(\rho_{j-1})^2 } \\
    %
    &\leq \gamma_{\Delta p}\max_{j=0,...,i}|\Tilde{\chi}_j|
\end{align}
where we have exploited the relationship
\begin{align*}
    |\max_j \Delta\Tilde{p}_j| & \leq \max_j |\Delta\Tilde{p}_j| \\
    |\min_j \Delta\Tilde{p}_j| & \leq \max_j |\Delta\Tilde{p}_j|
\end{align*}
By applying the same methodology, is proven the inequality
\begin{equation}
    |\psi^i_{\Delta v}| \leq \gamma_{\Delta v}\max_{j=0,...,i}|\Tilde{\chi}_j|
\end{equation}
Then,
\begin{equation}\label{eq:psi_property}
    a|\psi^i_{\Delta p}| + b|\psi^i_{\Delta v}| \leq (a\gamma_{\Delta p}+b\gamma_{\Delta v})\max_{j=0,...,i}|\Tilde{\chi}_j|.
\end{equation}
%
%\textcolor{black}{See Appendix \ref{App_3} in \cite{Mirabilio2020CDC_Extended} for more details on the calculations in (\ref{eq:psi_property}).} 
%
Define 
\begin{equation}\label{eq:gammatilde_parameters}
    d = a\gamma_{\Delta p}+b\gamma_{\Delta v} > 0,  \ \ \Upsilon\in(0,1),
\end{equation}
then for the time derivative of $W_i$ in (\ref{eq:dot_LyapunovFunction_ith_2_inequality1}) the following holds
%
% By exploiting the variance property in (\ref{eq:variance_property}) and (\ref{eq:psi_property}):
%
\begin{align}\label{eq:dot_LyapunovFunction_ith_2_inequality2}
    \nonumber \dot{W}_i &\leq -\alpha|\Tilde{\chi}_i|^2 + d|\Tilde{\chi}_i| \max_{j=0,...,i-1}|\Tilde{\chi}_j| + \Upsilon\alpha|\tilde{\chi}_i|^2 - \Upsilon\alpha|\tilde{\chi}_i|^2 \\
    &\leq -(1-\Upsilon)\alpha|\Tilde{\chi}_i|^2, \quad \forall \ |\Tilde{\chi}_i| \geq \frac{d}{\alpha \Upsilon} \max_{j=0,...,i-1}|\Tilde{\chi}_j|.
\end{align}
Since $\alpha > 0$, the inequality in (\ref{eq:dot_LyapunovFunction_ith_2_inequality2}) satisfies the Input-to-State Stability (ISS) condition (see \cite{B_khalil_2002}). According to \cite[Theorem~4.19]{B_khalil_2002}, the inequality in (\ref{eq:closed_loop_platoon_ISS_ineq2}) is verified. Moreover, 
\begin{equation}\label{eq:gammatilde}
    \gamma(s) = \Tilde{\gamma}s  \ \forall \ s \geq 0, \ \Tilde{\gamma} = \sqrt{\frac{\Bar{\alpha}}{\underline{\alpha}}}\frac{d}{\alpha \Upsilon} > 0. 
\end{equation}
Since the parameters $a,b \geq 0$ in the dynamics of $\rho_i$ in (\ref{eq:rho_dynamic_system}) can be arbitrarily selected, the constant $d$ defined in (\ref{eq:gammatilde_parameters}) can be chosen such that $\Tilde{\gamma}$ in (\ref{eq:gammatilde}) belongs to $(0,1)$.

\begin{flushright}
    $\square$
\end{flushright}

\subsection{Proof of Theorem \ref{thm:string_stability_control2}}\label{apx:string_stability_control2_proof}

The first part of the proof is based on the forward recursive application of the ISS property in Lemma \ref{lmm:closed_loop_ISS_control2} through an inductive method.
%The proof is based on the \textit{forward} recursive application of the ISS property in Lemma (\ref{lmm:closed_loop_ISS_control1}) \textcolor{blue}{ through an inductive method.}
\noindent   
For $i=0$:
\begin{align}
  |\Tilde{\chi}_0(t)| &\leq \beta(|\Tilde{\chi}_0(0)|,t), \ \forall \ t \geq 0.
\end{align}
For $i=1$:
\begin{equation}
    |\Tilde{\chi}_1(t)| \leq \beta(|\Tilde{\chi}_1(0)|,t)+\Tilde{\gamma}|\Tilde{\chi}_0(\cdot)|^{[0,t]}_\infty, \ \forall \ t \geq 0,
\end{equation}
where $|\Tilde{\chi}_0(\cdot)|^{[0,t]}_\infty \leq \beta(|\Tilde{\chi}_0(0)|,0)$. Defining $|\Tilde{\chi}_M(0)| = \max\{ |\Tilde{\chi}_0(0)| , |\Tilde{\chi}_1(0)| \}$, then for both $i=0$ and $i=1$:
\begin{align}
    |\Tilde{\chi}_0(t)| &\leq \beta(|\Tilde{\chi}_M(0)|,0), \ \forall \ t \geq 0, \\
    |\Tilde{\chi}_1(t)| &\leq \beta(|\Tilde{\chi}_M(0)|,0)(1+\Tilde{\gamma}), \ \forall t \ \geq 0.
\end{align}
For $i=2$:
\begin{align}
    \nonumber |\Tilde{\chi}_2(t)| &\leq \beta(|\Tilde{\chi}_2(0)|,t) \\
    & \quad +\Tilde{\gamma}\max_{j=0,1}|\Tilde{\chi}_j(\cdot)|^{[0,t]}_\infty, \ \forall \ t \geq 0.
\end{align}
Defining $|\Tilde{\chi}_M'(0)| = \max_{j=0,1,2}\{ |\Tilde{\chi}_j(0)| \}$, since $\Tilde{\gamma} > 0$, then
\begin{align}
    |\Tilde{\chi}_0(t)| \leq \beta(|\Tilde{\chi}_M'(0)|,0)(1+\Tilde{\gamma}), \ \forall \ t \geq 0, \\
    |\Tilde{\chi}_1(t)| \leq \beta(|\Tilde{\chi}_M'(0)|,0)(1+\Tilde{\gamma}), \ \forall \ t \geq 0 ,
\end{align}
and
\begin{equation}
    |\Tilde{\chi}_2(t)| \leq \beta(|\Tilde{\chi}_M'(0)|,0)(1+\Tilde{\gamma}+\Tilde{\gamma}^2), \ \forall t \geq 0.
\end{equation}
By recursively applying these steps, and since $\Tilde{\gamma}\in(0,1)$ for hypothesis, for each $i\in\mathcal{I}_N^0$ we state:
\begin{align}
    \nonumber |\Tilde{\chi}_i(t)| &\leq \beta\left( \max_{j=0,...,i}|\Tilde{\chi}_j(0)|,0 \right)\sum_{j=0}^i \Tilde{\gamma}^j \\
    \nonumber &\leq \beta\left( \max_{j=0,...,i}|\Tilde{\chi}_j(0)|,0 \right)\sum_{j=0}^\infty \Tilde{\gamma}^j \\
    &\leq \frac{1}{1-\Tilde{\gamma}}\beta\left( \max_{j=0,...,i}|\Tilde{\chi}_j(0)|,0 \right), \ \forall \ t \geq 0.\label{eq_Chi_bounded}
\end{align}
Then
\begin{equation}\label{eq:simple_string_stability_control1}
    \max_{i\in\mathcal{I}_N^0}|\Tilde{\chi}_i(t)| \leq \frac{1}{1-\Tilde{\gamma}}\beta\left( \max_{i\in\mathcal{I}_N^0}|\Tilde{\chi}_i(0)|,0 \right), \ \forall \ t \geq 0.
\end{equation}
Define $\omega(s) = \beta(s,0), \ s \geq 0$. By definition of $\mathcal{KL}$ functions, $\omega$ is $\mathcal{K}_\infty$ and hence invertible. Since (\ref{eq:simple_string_stability_control1}) holds for any $t \geq 0$, then 
\begin{equation}\label{eq_delta_forChi_bounded}
    \delta = \omega^{-1}((1-\Tilde{\gamma})\epsilon), \quad \forall \ \epsilon \geq 0.
\end{equation}
The value of $\delta$ in (\ref{eq_delta_forChi_bounded}) does not depend on the system dimension. From (\ref{eq_Chi_bounded}), (\ref{eq:simple_string_stability_control1}) and (\ref{eq_delta_forChi_bounded}), String Stability is ensured according to Definition \ref{def:TSS_string_stability}.
%\end{proof} 
%\noindent

We focus now on the possibility to ensure Asymptotic String Stability.
%
%\begin{theorem}\label{thm:asymptotic_string_stability_control2}
%    Closed loop system described by (\ref{eq:car_following_dynamics_0_closed_loop}) and (\ref{eq:car_following_dynamics_i_closed_loop_2}), satisfying the String Stability condition in Theorem \ref{thm:string_stability_control2}, is Asymptotically String Stable.
%\end{theorem}
%
%\begin{proof}
This second part of the proof is based on a composition of Lyapunov functions (see \cite{B_khalil_2002}). We consider the function $W_i$ associated with the $i$-th dynamical system, for $i\in\mathcal{I}_N^0$, that is described in (\ref{eq:LyapunovFunction_ith_2}) and  satisfies the condition in (\ref{eq:W_i_bounded}).

Let us consider the time derivative of $W_i$ in (\ref{eq:dot_LyapunovFunction_ith_2_inequality1}).
\begin{color}{black}
    Since we consider the dynamics in (\ref{eq:closed_loop_platoon_short_0}) and (\ref{eq:closed_loop_platoon_short_i}) with respect to $\Tilde{\chi}_i = \hat{\chi}_i-\chi_{e,i}$, we define $\Delta\Tilde{p}_i = \Delta p_i -  \Delta\Bar{p}$, then for the macroscopic functions $\psi^i_{\Delta p}$ and $\psi^i_{\Delta v}$, the following inequalities are proved:
    \begin{align}
        \nonumber |\psi^i_{\Delta p}| &\leq \gamma_{\Delta p}\sqrt{\sigma^2_{\Delta p,i}} \\
        \nonumber &= \gamma_{\Delta p}\left( \frac{1}{i+1}\sum_{j=0}^{i} \Delta\Tilde{p}_j^2 - \frac{1}{(i+1)^2}\left( \sum_{j=0}^{i} \Delta\Tilde{p}_j \right)^2  \right)^{\frac{1}{2}} \\
        \nonumber &\leq \gamma_{\Delta p}\frac{1}{\sqrt{i+1}}\left( \sum_{j=0}^{i} \Delta\Tilde{p}_j^2 \right)^{\frac{1}{2}} \\
        &\leq \gamma_{\Delta p}\frac{1}{\sqrt{i+1}} \sum_{j=0}^{i}|\Delta\Tilde{p}_j|
    \end{align}
    where we have used the inequality $|x|_2 \leq |x|_1$. In the same way we can prove that 
    \begin{equation}
        |\psi^{i}_{\Delta v}| \leq \gamma_{\Delta v}\frac{1}{\sqrt{i+1}} \sum_{j=0}^{i} |\Delta v_j|
    \end{equation}
    Then,
    \begin{align}
        \nonumber a|\psi^i_{\Delta p}| &+ b|\psi^i_{\Delta v}| \leq \frac{1}{\sqrt{i+1}}\left( a\gamma_{\Delta p} \sum_{j=0}^{i}|\Delta\Tilde{p}_j| + b\gamma_{\Delta v} \sum_{j=0}^{i}|\Delta v_j| \right) \\
        \nonumber &= \frac{1}{\sqrt{i+1}} \sum_{j=0}^i \left| \left[
        \begin{array}{cccc}
            a\gamma_{\Delta p} & 0 & 0 & 0 \\
            0 & b\gamma_{\Delta v} & 0 & 0 \\
            0 & 0 & 0 & 0 \\
            0 & 0 & 0 & 0
        \end{array} \right] 
        % \left[ 
        % \begin{array}{c}
        %     \Delta\Tilde{p}_j  \\
        %     \Delta v_j \\
        %     \rho_{j-1}
        % \end{array}
        % \right]
        \Tilde{\chi}_j
        \right|_1 \\
        &\leq \frac{2}{\sqrt{i+1}}\max\{a\gamma_{\Delta p},b\gamma_{\Delta v}\} \sum_{j=0}^i |\Tilde{\chi}_j|
    \end{align}
    %
    % where 
    % %
    % \begin{align}
    %     %
    %     % \Tilde{k}_{i} &= \lambda > 0 , \\
    %     %
    %     \Tilde{k}_{i} &= \sqrt{\frac{3}{i+1}}\max\{a\gamma_{\Delta p},b\gamma_{\Delta v}\} > 0.
    %     %
    % \end{align}
    % %
\end{color}
%
% \textcolor{red}{for $|[\chi_{i-1} \ ...\ \chi_0]| < r$} and non-negative constants $\Tilde{k}_{ij}$.
Let $\hat{\chi}$ and $\hat{\chi}_e$ be the extended lumped state of the platoon and the extended equilibrium point respectively, defined in a similar way as (\ref{eq:platoon_state}) and (\ref{eq:platoon_equilibrium}). Let us consider $\Tilde{\chi} = \hat{\chi}-\hat{\chi}_e$ and the parameters $d_i>0$ to define a composite function $ W_c(\Tilde{\chi})$:
\begin{equation}\label{eq:composite_LyapunovFunction_2}
    W_c(\Tilde{\chi}) = \sum_{i=0}^N d_i W(\Tilde{\chi}_i) .
\end{equation}
It clearly verifies
\begin{equation}\label{eq:W_bounded}
    \underline{\alpha}_c|\Tilde{\chi}|^2 \leq W_c(\Tilde{\chi}) \leq \Bar{\alpha}_c|\Tilde{\chi}|^2
\end{equation}
where
\begin{equation}\label{eq:W_bounded_alphas}
    \underline{\alpha}_c = \min_{i\in\mathcal{I}_N^0} \{d_i\}\underline{\alpha}, \ \ \Bar{\alpha}_c=\max_{i\in\mathcal{I}_N^0}\{d_i\}\Bar{\alpha}.
\end{equation}
The time derivative of $W_c$ in (\ref{eq:composite_LyapunovFunction_2}) satisfies the inequality
\begin{equation}\label{eq:composite_LyapunovFunction_dot_2}
    \dot{W}_c(\Tilde{\chi}) \leq \sum_{i=0}^N d_i \left[ -\alpha |\Tilde{\chi}_i|^2 + \sum_{j=0}^{i-1}\Tilde{k}_{i} |\Tilde{\chi}_j||\Tilde{\chi}_i| \right].
\end{equation}
where 
\begin{align}
    \Tilde{k}_{0} &= 0 , \\
    \Tilde{k}_{i} &= \frac{2}{\sqrt{i}}\max\{a\gamma_{\Delta p},b\gamma_{\Delta v}\} > 0, \ i \geq 1.
\end{align}
\begin{color}{black}
    Let us introduce the operator  $\phi:\mathbb{R}^{2N+1}\rightarrow\mathbb{R}^{N+1}$, defined as
    \begin{equation}\label{eq:phi_definition}
        \phi(\Tilde{\chi}) = [ |\Tilde{\chi}_0|,|\Tilde{\chi}_1|,...,|\Tilde{\chi}_N| ]^T.
    \end{equation}
\end{color}
% \begin{color}{blue}
%     Let us introduce $\phi:\mathbb{R}^{2N+1}\rightarrow\mathbb{R}^{N+1}$ defined below
%     %
%     \begin{equation}\label{eq:phi_definition}
%         \phi(\Tilde{\chi}) = [ |\Tilde{\chi}_0|,|\Tilde{\chi}_1|,...,|\Tilde{\chi}_N| ]^T.
%     \end{equation}
%     %
% \end{color}
%and 
%
% \begin{equation}
%     \phi = [ \phi_1,...,\phi_N ]^T
% \end{equation}
%
Then, equation (\ref{eq:composite_LyapunovFunction_dot_2}) can be rewritten as
\begin{equation}\label{eq:Wdot_bounded2}
    \dot{W}_c(\Tilde{\chi}) \leq -\frac{1}{2}\phi(\Tilde{\chi})^T(DS+S^TD)\phi(\Tilde{\chi})
\end{equation}
where
\begin{equation}
    D =  diag(d_0,d_1,...,d_N)
\end{equation}
and $S$ is an $N\times N$ matrix whose elements are
\begin{equation}
    s_{ij} = 
    \begin{cases}
        \alpha \:\: & i = j \\
        -\Tilde{k}_{i} \:\: & i < j \\
        0 \:\: & i > j
    \end{cases}
\end{equation}
For $\alpha > 0$, each leading principal minor of $S$ is positive and hence it is an $M$-matrix. By \cite[Lemma~9.7]{B_khalil_2002} there exists a matrix $D$ such that $DS+S^TD$ is positive definite. Consequently, $\dot{W}_c$ in (\ref{eq:Wdot_bounded2}) is negative definite. 
\begin{color}{black}
    It follows that $W_c$ in (\ref{eq:composite_LyapunovFunction_2}) is a Lyapunov function for the overall platoon system described by %(\ref{eq:car_following_dynamics_0_closed_loop}) and
    (\ref{eq:car_following_dynamics_i_closed_loop_2}). Therefor, there exists a $\mathcal{KL}$ function $\beta_c : \mathbb{R}^+\times\mathbb{R}^+ \rightarrow \mathbb{R}^+$ such that
    \begin{equation}\label{eq:beta_c_inequality2}
        |\Tilde{\chi}(t)| \leq \beta_c(|\Tilde{\chi}(0)|,t), \:\: \forall \ t \geq 0.
    \end{equation}
    Condition in  (\ref{eq:beta_c_inequality2}) ensures the asymptotic stability:
    \begin{equation}\label{eq:limit_condition_foralli2}
        \lim_{t\rightarrow\infty}|\Tilde{\chi}_i(t)| = 0, \quad \forall \ i\in\mathcal{I}_N^0.
    \end{equation}
    The platoon system is proved to be String Stable by  (\ref{eq:simple_string_stability_control1}) and (\ref{eq_delta_forChi_bounded}). Consequently, for each $i\in\mathcal{I}_N^0$ the state evolution $|\Tilde{\chi}_i|$ is constrained by a bound that is independent from the system dimension. Furthermore, from (\ref{eq:limit_condition_foralli2}) Asymptotic String Stability is ensured according to Definition \ref{def:ATSS_string_stability}.
    
    % can be rewritten as:
    % %
    % \begin{equation}
    %     \max_{i\in\mathcal{I}_N^0}|\Tilde{\chi}_i(t)| \leq \beta_c(|\Tilde{\chi}(0)|,t), \ \forall \ t \geq 0
    % \end{equation}
    % %
    
    % \begin{flushright}
    %     $\square$
    % \end{flushright}
    
\end{color}

\end{document}